\newif\ifonecol

\documentclass[journal,onecolumn]{IEEEtran}%

\IEEEoverridecommandlockouts

\usepackage{amssymb}
\usepackage{bm}
\usepackage{subfigure}
\usepackage{tikz}
\usepackage{pstricks}
\usepackage{algorithm}
\usepackage{algorithmic}
\usepackage[reqno]{amsmath}
\usepackage{amsthm}
\usepackage{graphics,graphicx}
\usepackage{fancyhdr}
\usepackage{xcolor}
\usepackage{amsfonts}
\usepackage{bbm}
\usepackage{setspace}
\usepackage{epstopdf}
\usepackage{enumerate}
\usepackage{url}

\usepackage{soul}
\usepackage{todonotes}

\usepackage[normalem]{ulem}
\usepackage{soul}
\usepackage{cancel} 
\newcommand{\inma}[1]{#1}

\providecommand{\U}[1]{\protect\rule{.1in}{.1in}}
\newtheorem{theorem}{Theorem}

 \newtheorem{definition}{Definition}
 \newtheorem{remark}{Remark}

\title{Stable Wireless Network Control Under Service Constraints}
   
 \author{
  \IEEEauthorblockN{Martin Kasparick\IEEEauthorrefmark{1} and Gerhard Wunder\IEEEauthorrefmark{2}} 
   \\\IEEEauthorblockA{\IEEEauthorrefmark{1} Fraunhofer Heinrich Hertz Institute, Berlin   
   (\textit{martin.kasparick@hhi.fraunhofer.de})} 
   \\\IEEEauthorblockA{\IEEEauthorrefmark{2} Freie Universit\"at Berlin
   (\textit{g.wunder@fu-berlin.de})}  
 }

\begin{document}

\ifonecol
\onehalfspacing
\fi

\maketitle

\begin{abstract}
We consider the design of wireless queueing network control policies with
particular focus on combining stability with additional application-dependent
requirements. Thereby, we consequently pursue a cost function based approach
that provides the flexibility to incorporate constraints and requirements of
particular services or applications. As typical examples of such requirements,
we consider the reduction of buffer underflows in case of streaming traffic,
and energy efficiency in networks of battery powered nodes. Compared to the
classical throughput optimal control problem, such requirements significantly
complicate the control problem. We provide easily verifyable theoretical
conditions for stability, and, additionally, compare various candidate cost
functions applied to wireless networks with streaming media traffic. Moreover,
we demonstrate how the framework can be applied to the problem of energy
efficient routing, and we demonstrate the aplication of our framework in
cross-layer control problems for wireless multihop networks, using an advanced
power control scheme for interference mitigation, based on successive convex
approximation. In all scenarios, the performance of our control framework is
evaluated using extensive numerical simulations.

\end{abstract}

\onehalfspacing

\section{Introduction}

\inma{Since} the seminal work of
Tassiulas and Ephremides \cite{Tassiulas:1992}, throughput optimal control of
stochastic networks by MaxWeight-type policies has gained considerable
attention. However, \inma{soon} it was
 discovered that MaxWeight can lead to
significant delays \cite{subra_07_deco}, which may render
 \inma{the} practical application
\inma{of such control strategies} impossible.
\inma{As a result,} significant efforts followed
 \inma{in order to reduce} the
delay of MaxWeight-type policies \cite{Bui2011,Ying2010,Naghshvar2012}.

 \inma{Notably}, delay is not the
only relevant performance measure. For example, in networks
\inma{with size-limited}  receive
buffers, a policy that only considers
delay can be even harmful, since it \inma{might} be better to absorb
traffic peaks at intermediate buffers instead of routing packets as fast as
possible towards the end user. Moreover, when  the offered performance \inma{already meets application
targets}, it is unnecessary to invest further in incremental improvements of
particular performance measures. \inma{Consequently}, a
sophisticated control approach should be able to flexibly cope with specific,
possibly time-varying, requirements that are dictated by the actual
applications (streaming services, gaming, etc.) \cite{blau2007cost}. \inma{As an example,
consider} multimedia streaming traffic, which
constitutes a  large fraction of the traffic observed in wireless
networks nowadays. Here, we have to ensure that application buffers do not
 \inma{drain} in order to avoid
interruptions of the stream. This results in \textit{minimum buffer size
constraints}. \inma{On the other hand}, it should be  avoided
that buffers grow too large, since (1) buffers could be limited in size and
(2) the user might want to switch the video stream, which implies a waste of
already invested network resources.

\inma{Recently} in \cite{Wunder:2012}, we introduced a control framework
called $\mu$-MaxWeight which is specifically suited  to
\inma{address} \inma{such} \inma{policy design} challenges
\inma{in a joint fashion}. In general, our approach
\inma{not only} guarantees stability but  \inma{it} additionally allows to
incorporate certain application requirements \inma{derived from an
underlying cost function}. Moreover, it can flexibly react to changing user
requirements by applying different cost functions. Based on this
framework, we approach the design of wireless networks and corresponding
control policies from a service centric --or application centric-- point of
view. This may also imply a change of viewpoint: from stabilizing transmit
buffers to controlling the application buffers at the user side. We assume
that the service layer constraints can be expressed as requirements on buffer
states, i.e., applications may require maximum and minimum buffer states
(which, however, may differ for different network entities). While the first
type of constraints is well investigated, the second type bears serious
challenges, especially in multihop networks.

\subsection{Contributions}

We provide general sufficient conditions on the stability of a $\mu$-MaxWeight
based control policy, and \inma{we} prove two simpler stability
results \inma{that are} based on additional assumptions on the
network model. Moreover, we discuss several 
\inma{exemplary} field\inma{s} of application for our
control framework.

{In the first application,}  we
show that with a sophisticated cost function choice, buffer sizes of
particular queues can be steered towards beneficial operating points. The
choice of the underlying cost function is crucial for the performance of the
resulting control policy. Therefore, to numerically compare different cost
functions, we define the notion of \textit{queue outage}, \inma{which} 
\inma{both} incorporates buffer overflows to capture effects of
high packet holding costs or even dropped packets \inma{as well
as} buffer underflows to capture service interruptions.

In  {the} second application, we use
our framework to design a control policy for an energy constrained wireless
network, where traffic is routed via battery limited intermediate nodes. We
assume that when a node has no traffic in each of its queues, it can switch to
an energy-saving mode. Therefore, the goal is to maximize the ``down time'' of
as many nodes as possible, while guaranteeing stability of the network.

{Eventually, in the third application,}
 we consider cross-layer
optimization of wireless networks, and {we} combine
network-layer control with physical layer interference mitigation.
{Thereby, we}  demonstrate how our
framework can be combined with physical layer power control based on
successive convex approximation (SCA).

\subsection{Related Work}

After the seminal work in \cite{Tassiulas:1992}, significant research
activities followed. Many focussed on delay reduction of backpressure based
policies, such as \cite{Bui2011}. A general class of throughput optimal
policies with improved delay performance \inma{has been} recently
presented in \cite{Naghshvar2012}, and a general survey on recent policy
synthesis techniques can be found, for example, in \cite{MeynBook2007}.

Considering queueing networks, previous works consider underflow constraints
mainly in broadcast or simpler networks \cite{Cioffi2010}. For example, in
\cite{Cioffi2010}, the problem is investigated in a network of multiple
transmitter-receiver pairs with cooperating transmitters. However, since it
requires that buffer levels for each user progress independently of other
users, it cannot be easily generalized to the case of arbitrary (multihop)
networks (not even to a simple tandem network).

Recently, there has been some interest in queueing network control with an
arbitrary underlying cost-metric \cite{Meyn:2007}, allowing to incorporate
application-dependent constraints in the control policy. However, stability
and also cost performance crucially depend on the parameter choice,
{and this parameter choice} 
has to be {found}  in simulations.

The MaxWeight policy is also frequently applied in wireless cross-layer
control problems, in particular, {in} joint routing and power
control {problems} (see, e.g., \cite{Georgiadis_06}). A common
assumption on the physical layer is a fully coupled wireless network, where
nodes treat interference as noise. In this case, power control results in a
weighted sum-rate maximization problem which is known to be NP-hard
\cite{Luo2008}. In \cite{Papandriopoulos2009}, an SCA method was introduced as
a means to treat the inherent non-convexity.  
Compared to traditional \inma{power control in}  wireless networks  \cite{Kasparick2014},  it
eliminates the need to use the common \inma{unrealistic} high-SINR
approximation together with a log-transform of variables (see, e.g.,
\cite{Xi2010}). In \cite{Matskani2012_TSP},  a cross-layer
control formulation \inma{was used} based on the MaxWeight technique.

\subsection{Organization}

The rest of the paper is organized as follows. The system model is described
in Section \ref{sec:sysmodel}, and important preliminaries are summarized in
Section \ref{sec:preliminaries}. The $\mu$-MaxWeight policy, together with
important stability results, is introduced in Section \ref{sec:stability}.
{In} Section \ref{sec:costfunctions}, {we}
treat the important issue of cost-function choice, and {we}
show, how  {an appropriate choice}
can be used to steer buffers towards certain target values. In Section
\ref{sec:applications}, we apply our framework in different scenarios, such as
multimedia traffic, energy constrained nodes, and cross-layer network control.
In Section \ref{sec:conclusion}, we draw some relevant conclusions.

\subsection{Notation}

We use boldface letters to denote vectors as well as matrices, and common
letters with subscript are the elements, such that $A_{i}$ is the $i$-th
element of vector $\bm{A}$ and $B_{ij}$ is the element in row $i$ and column
$j$ of matrix $\bm{B}$. Moreover, $\bm{A}^{T}$ refers to the transpose of
$\bm{A}$. $\mathbb{E}\{X\}$ denotes the expected value of random variable $X$.
Let $\bm{I}$ denote the identity matrix of appropriate dimension. Furthermore,
we denote $\bm{1}$ the vector of all ones. $\mathbf{diag}(a_{1},a_{2},...)$
refers to a diagonal matrix built from the elements $a_{1},a_{2},...$,
$\Vert\cdot\Vert_{i}$ denotes the $l_{i}$ vector norm, and $\Vert
\mathbf{x}\Vert$ is an arbitrary norm. Furthermore, we use $\mathcal{A}^{c}$
to denote the complement of a set $\mathcal{A}$. The probability of
$\mathcal{A}$ is denoted as $\Pr\{\mathcal{A}\}$. $\left\langle \cdot
,\cdot\right\rangle $ denotes the scalar product of two vectors. The indicator
$\mathbb{I\{\cdot\}}$ {takes the value}
 1 if the argument is true, and
{it takes the value}  0 otherwise.

\section{System Model}

\label{sec:sysmodel}

In order to model queueing networks, we use a simple, discrete time,
stochastic network model, called Controlled Random Walk (CRW) model. We
consider a queueing network with $m$ queues in total that represent $m$
physical buffers with unlimited storage capacity. We arrange the queue backlog
in the vector $\mathbf{Q}$, such that $\mathbf{Q}=\left[  {Q}_{1},\ldots
,{Q}_{m}\right]  ^{T}$, which we refer to as the queue state. Let
$\mathcal{M}$ be the set of queue indices. Suppose that the evolution of the
queueing system is time slotted with $t\in\mathbb{N}_{0}$. Then, our model is
defined by the queueing law:%
\begin{equation}
\mathbf{Q}\left(  t+1\right)  =\left[  \mathbf{Q}\left(  t\right)
+\mathbf{B}\left(  t+1\right)  \mathbf{U}\left(  t\right)  \right]
^{+}+\mathbf{A}\left(  t+1\right) , \label{eqn:myCRW}%
\end{equation}
where $[{x}]^{+}:=\max\{0,x\}$. The vector process $\mathbf{A}\left(
t\right)  \in\mathbb{N}_{0}^{m}$ (vector of arrival rates in packets per slot)
is the (exogenous) influx to the queueing system with mean $\bm{\alpha}\in
\mathbb{R}_{+}^{m}$\footnote{
{in the numerical evaluations in this paper, we assume Poisson distributed arrival processes}},
and $\mathbf{B}\left(  t\right)  \in\mathbb{Z}_{0}^{m\times l}$ is a matrix
process with average $\mathbf{B}\in\mathbb{Z}_{0}^{m\times l}$, where
$l\in\mathbb{Z}_{+}$ is the number of control decisions to be made (that is,
the number of links in the network). The matrix $\mathbf{B}\left( t\right) $
contains both information about the network topology (that is, about
connectivity or possible routing paths in between queues) and service rates
along the specific links. We assume that the control vector $\mathbf{u}%
=\mathbf{U}\left(  t\right)  $ in slot $t$ is an element of the set $\left\{
0,1\right\}  ^{l}$. Moreover, we can impose further (linear) constraints on
the control using the binary constituency matrix $\mathbf{C}\in\mathbb{Z}%
_{0}^{l_{m}\times l}$ (with $l_{m}>0$ being the number of constraints), where
each row of $\mathbf{C}$ corresponds to a particular controllable network
resource.  More precisely, we require that  $\bm{Cu}\leq\bm{1}$, such that
$\bm{u}(t)\in\left\{ \bm{u}\in\{0,1\}^{l}: \bm{Cu}\leq\bm{1}\right\} $. For
the sake of notational simplicity, we omit the time index in the following if
possible. Throughout the entire paper, $\bm{x}\in\mathbb{N}_{0}^{m}$ denotes
the instantaneous backlog.

In the following, the queueing system (\ref{eqn:myCRW}) is assumed to be a
$\delta_{\mathbf{0}}$-irreducible Markov chain (with $\delta_{\mathbf{0}}$
being the point measure at $\bm{x}=\mathbf{0}$).

\subsection{Stability}

Throughout this paper we use the following definition of stability.

\begin{definition}
\label{definition_f_stable} A Markov chain is called \emph{f-stable}, if there
is an unbounded function $f:\mathbb{R}_{+}^{m}\rightarrow\mathbb{R}_{+}$, such
that for any $0<B<+\infty$ the set $\mathcal{B}:=\left\{  \mathbf{x}:f\left(
\mathbf{x}\right)  \leq B\right\}  $\ is compact, and furthermore it holds
\begin{equation}
\limsup_{t\rightarrow+\infty}\mathbb{E}\left\{  f\left(  \mathbf{Q}\left(
t\right)  \right)  \right\}  <+\infty. \label{condition_f_stable}%
\end{equation}

\end{definition}

In Definition \ref{definition_f_stable}, the function $f$ is unbounded in all
positive directions so that $f\left(  \mathbf{x}\right)  \rightarrow\infty$ if
$\Vert\mathbf{x}\Vert\rightarrow\infty$, \inma{regardless of the actual path
taken in $\mathbb{R}_{+}^{m}$}. Choosing directly $f\left(  \mathbf{x}\right)
=\Vert\mathbf{x}\Vert$, Definition \ref{definition_f_stable} is equivalent to
the definition of \emph{strongly stable} \cite{Leonardi:2003}, which implies
weak stability. Clearly, for any $f\left(  \mathbf{x}\right)  $ that grows
faster than $\Vert\mathbf{x}\Vert$, inequality (\ref{condition_f_stable})
implies that the Markov chain is strongly stable.

\inma{Using f-stability, we can define our notion of throughput optimality}.
For this we will call a vector of arrival rates $\boldsymbol{\alpha}%
\in\mathbb{R}_{+}^{m}$ \emph{stabilizable}, when the corresponding queueing
system, driven by some specific scheduling policy, is f-stable.

\begin{definition}
\label{definition_throughput_optimal} A scheduling policy is called
\emph{throughput-optimal} if it keeps the Markov chain \inma{f-stable} for
any vector of arrival rates $\boldsymbol{\alpha}$ which is stabilizable by
some policy.
\end{definition}

\section{Preliminaries}

\label{sec:preliminaries}

Let us introduce a cost function%
\[
c:\mathbb{N}_{0}^{m}\rightarrow\mathbb{R}_{+},\mathbf{x}\hookrightarrow
c\left(  \mathbf{x}\right)  ,
\]
assigning any queue state a non-negative number. Typically, the goal is to
minimize the average cost over a given finite or infinite time period, or some
discounted cost criterion. The optimal solution to the resulting problems
(which in discrete time can be modeled as a \emph{Markov Decision Problem})
can be found by dynamic programming, which is, however, infeasible for large
networks. A simple approach to queueing network control is the \emph{myopic or
greedy policy}. Such a policy selects the control decision that minimizes the
expected cost only for the next time slot.

In \cite{Meyn:2007}, a cost function based policy design framework called
$h$-MaxWeight is introduced, which is a generalization of the MaxWeight
policy. Meyn considers a slightly different definition of the CRW model, which
is characterized by
\begin{equation}
\mathbf{Q}\left(  t+1\right)  =\mathbf{Q}\left(  t\right)  +\mathbf{B}\left(
t+1\right)  \mathbf{U}\left(  t\right)  +\mathbf{A}\left(  t+1\right)
.\label{eqn:CRW}%
\end{equation}
The control {vector} $\mathbf{U}\left(  t\right)
\in\mathbb{N}_{0}^{l}$ is an element of the region $\mathcal{U}^{\ast}\left(
\mathbf{x}\right)  :=\mathcal{U}\left(  \mathbf{x}\right)  \cap\left\{
0,1\right\}  ^{l}$, with
\[
\mathcal{U}\left(  \mathbf{x}\right)  :=\left\{  \mathbf{u}\in\mathbb{R}%
_{+}^{l}:\mathbf{Cu}\leq\bm{1},\left[  \mathbf{Bu}+\bm{\alpha}\right]
_{i}\geq0\text{ for }x_{i}=0\right\}  .
\]
A $h$-MaxWeight policy {based on} \cite{Meyn:2007} chooses the
control vector according to%
\begin{equation}
\underset{\bm{u}\in\mathcal{U}^{\ast}\left(  \mathbf{x}\right)  }{\arg\min
}\langle\nabla h(\bm{x}),\bm{Bu}+\bm{\alpha}\rangle.\label{eq:hMW}%
\end{equation}
Thus, the policy is myopic with respect to the gradient of some perturbation
$h$ of the underlying cost function. Meyn develops two main constraints on the
function $h$. The first {constraint} requires the partial
derivative of $h$ to vanish when queues become empty, i.e.,
\begin{equation}
\frac{\partial h}{\partial x_{i}}(\bm{x})=0,\text{ if }x_{i}%
=0.\label{eq:derivcond}%
\end{equation}
Moreover, the dynamic programming inequality $\min_{\bm{u}\in\mathcal{U}%
(\bm{x})}\langle\nabla h(\bm{x}),\bm{Bu}+\bm{\alpha}\rangle\leq-c(\bm{x})$ has
to hold for the function $h$. When $h$ is non-quadratic, the derivative
condition (\ref{eq:derivcond}) is not always fulfilled. Therefore a
perturbation technique is used, where $h(\bm{x})=h_{0}(\tilde{\bm{x}})$, hence
it is a perturbation of a function $h_{0}$. Two perturbations are proposed: an
exponential perturbation with $\theta\geq1$, given by
\[
\tilde{x}_{i}:=x_{i}+\theta\left(  e^{-\frac{x_{i}}{\theta}}-1\right)  ,
\]
and a logarithmic perturbation with $\theta>0$, defined as%
\begin{equation}
\tilde{x}_{i}:=x_{i}\log\left(  1+\frac{x_{i}}{\theta}\right)
.\label{eq:pertLog}%
\end{equation}
While the first approach shows better performance in simulations, the
stability of the resulting policy depends on the parameter $\theta$ being
sufficiently large (determined by the considered network setting). This is
overcome by the second perturbation, which is stabilizing for each feasible
$\theta$. However, it comes with the additional constraint
\begin{equation}
\frac{\partial h_{0}}{\partial x_{i}}(\bm{x})\geq\epsilon x_{i},\quad\forall
i\in\mathcal{M},\label{eq:CondLog}%
\end{equation}
which is a significant limitation on the space of functions that can be chosen
as $h_{0}$ (\inma{apart from the additional condition that $\nabla h_{0}$
has to be Lipschitz-continuous}). \inma{Moreover, we would like to emphasize that,
beyond these specific two examples, there is no general framework that can be
applied to explore the stability properties of other choices of perturbations.
This is now provided with the $\mu$-MaxWeight framework.}

\section{$\mu$-MaxWeight}

\label{sec:stability}

In what follows, we consider {and investigate} scheduling
policies of the form%
\begin{equation}
\mathbf{u}^{\ast}(\mathbf{x})=\underset{{\mathbf{u}}\in\mathbb{R}_{+}%
^{n}:\bm{Cu}\leq\bm{1}}{\arg\min}\langle\boldsymbol{\mu}\left(  \mathbf{x}%
\right)  ,{\mathbf{Bu}+\bm{\alpha}}\rangle,\label{eqn:muMW}%
\end{equation}
where $\boldsymbol{\mu}:\mathbb{R}_{+}^{m}\rightarrow\mathbb{R}_{+}^{m}$ is a
vector valued function, and $\boldsymbol{\mu}(\bm{x})$ is called the
\emph{weight vector} for an instantaneous queue state $\mathbf{x}\in
\mathbb{R}_{+}^{m}$. Note that $\bm{\mu}$ is reminiscent of a vector field,
and it can thus be interpreted as a \emph{scheduling field} for which we
{subsequently} present a stability characterization. By
construction of the policy, we can, without loss of generality, normalize the
weight vector as%
\begin{equation}
\boldsymbol{\bar{\mu}}(\mathbf{x}):=\frac{\boldsymbol{\mu}(\mathbf{x}%
)}{\left\Vert \boldsymbol{\mu}(\mathbf{x})\right\Vert _{1}}%
,\label{eqn:normalize}%
\end{equation}
and hence $\left\Vert \boldsymbol{\bar{\mu}}(\mathbf{x})\right\Vert _{1}=1$.
Furthermore, we assume that the resulting policy is non-idling, i.e.,
$\left\Vert \boldsymbol{\mu}(\mathbf{x})\right\Vert _{1}=0$ if and only if
$\mathbf{x}=\mathbf{0}$. Below, we
{provide} generalized sufficient
conditions for throughput optimality of the systems \eqref{eqn:myCRW} and
\eqref{eqn:CRW} under control policy \eqref{eqn:muMW}.

\begin{theorem}
\label{thm:theorem1} Consider the queueing system (\ref{eqn:myCRW}), driven by
the control policy (\ref{eqn:muMW}) with some scheduling field
$\boldsymbol{\mu}$. The policy is throughput optimal (\inma{with respect to
the Definition \ref{definition_throughput_optimal}}) if the corresponding
normalized scheduling field given in \eqref{eqn:normalize} fulfills the
following conditions:

\begin{enumerate}
\item[(A1)] Given any $0<\epsilon_{1}<1$ and $C_{1}>0$, there is some
$B_{1}>0$, so that for any $\Delta\mathbf{x\in}\mathbb{R}^{m}$ with
$\left\Vert \Delta\mathbf{x}\right\Vert \mathbf{<}C_{1}$, we have $\left\vert
\bar{\mu}_{i}\left(  \mathbf{x}+\Delta\mathbf{x}\right)  -\bar{\mu}_{i}\left(
\mathbf{x}\right)  \right\vert \leq\epsilon_{1}$, for any $\mathbf{x}%
\in\mathbb{R}_{+}^{m}$ with $\left\Vert \mathbf{x}\right\Vert >B_{1}$, and
$\forall i\in\mathcal{M}$.

\item[(A2)] Given any $0<\epsilon_{2}<1$ and $C_{2}>0$, there is some
$B_{2}>0$, so that for any $\mathbf{x}\in\mathbb{R}_{+}^{m}$ with $\left\Vert
\mathbf{x}\right\Vert >B_{2}$ and $x_{i}<C_{2}$, we have $\bar{\mu}%
_{i}(\mathbf{x})\leq\epsilon_{2}$, for any $i\in\mathcal{M}$.
\end{enumerate}

Moreover, for any stabilizable arrival process, the queueing system is
f-stable under the given policy according to Definition
\ref{definition_f_stable}. The exact formulation of $f$ depends on the field
$\boldsymbol{\bar{\mu}}(\mathbf{x})$.
\end{theorem}

\begin{proof}
The proof is provided in full detail in \cite{WunderJournal12}.

\end{proof}

These conditions can be significantly simplified, assuming that
$B_{ij}(t)\geq-1$ (for each $i$,$j$ and $t$), and for each $j\in
\{1,\ldots,l\}$, there exists a unique value $i_{j}\in\{1,\ldots,m\}$
satisfying
\begin{equation}
B_{ij}(t)\geq0\quad\text{a.s. }\forall i\neq i_{j}\label{eq:MeynCondition}%
\end{equation}
(which is also assumed in Theorem 1.1 of \cite{Meyn:2007}).

\begin{theorem}
\label{thm:corollary1} Consider the queueing system (\ref{eqn:CRW}) driven by
the control policy (\ref{eq:hMW}) with some cost function $h$. Suppose the
corresponding scheduling field $\boldsymbol{\mu}(\bm{x}):=\nabla h(\bm{x})$ is
continuously differentiable, and Condition (\ref{eq:MeynCondition}) on the
network topology $\{\bm{B}(\cdot)\}$ holds. Then, the following conditions are
sufficient for throughput optimality:

\begin{enumerate}
\item[(C1)] For any $\epsilon>0$, there is some $C_{1}^{\ast}>0$, so that for
all $\left\Vert \bm{x}\right\Vert \geq C_{1}^{\ast}$ holds $\left\Vert
\nabla\log(\mu_{i}(\bm{x}))\right\Vert _{1}\leq\epsilon$ \inma{or the weaker
condition (i.e. it covers more policies) }$\left\Vert \nabla\mu_{i}%
(\bm{x})\right\Vert _{1}\leq\epsilon\left\Vert \boldsymbol{\mu}%
(\bm{x})\right\Vert _{1}$ for all $i\in\mathcal{M}$.

\item[(C2)] If $x_{i}=0$, then $\mu_{i}(\bm{x})=0,\;\forall i\in\mathcal{M}$.
\end{enumerate}
\end{theorem}

\begin{proof}
\inma{The proof can be found in Appendix \ref{sec:Appx1}.}
\end{proof}

\begin{theorem}
\label{thm:corollary2} Suppose, everything is as in \inma{Theorem \ref{thm:corollary1}}. Let the
scheduling field be defined as $\boldsymbol{\mu}(\bm{x}):=\nabla h_{0}%
(\tilde{\bm{x}})$ for some given simple perturbation $\tilde{\bm{x}}$. Then,
for some $\epsilon>0$,

\begin{enumerate}
\item[(D1)] $\frac{\partial\tilde{x_{i}}}{\partial x_{i}}\text{ is
Lipschitz,\ and }\frac{\partial\tilde{x_{i}}}{\partial x_{i}}\rightarrow
\infty,x_{i}\rightarrow\infty,$

\item[(D2)] $\frac{\partial h_{0}}{\partial x_{i}}\text{ is Lipschitz,\ and
}\frac{\partial h_{0}}{\partial\tilde{x_{i}}}\left(  \tilde{\bm{x}}\right)
\geq\left(  \frac{\partial\tilde{x_{i}}}{\partial x_{i}}\right)  ^{1+\epsilon
},x_{i}\rightarrow\infty,$
\end{enumerate}

is sufficient for stability.
\end{theorem}

\begin{proof}
\inma{
The proof can be found in Appendix \ref{sec:Appx2}.
}
\end{proof}

\inma{
\begin{remark}
(On necessary conditions) For the weaker part in C1) in Theorem
\ref{thm:corollary1} it was shown in \cite{ZhouISIT09}, in a wireless
broadcast setting, that the conditions are also necessary if the boundary of
the feasible (rate) region contains differentiable parts, i.e. parts where the
hyperplanes defined through the scheduling field are uniquely supported.
\end{remark}
}
\inma{
It is important to note that the conditions in Theorem \ref{thm:corollary1} and Theorem \ref{thm:corollary2}
cover indeed a larger class of throughput optimal policies compared to the
perturbation in (\ref{eq:pertLog}) as follows: Observe that for
(\ref{eq:pertLog}) we have $\frac{\partial\tilde{x_{i}}}{\partial x_{i}}%
=\log\left(  1+\frac{x_{i}}{\theta}\right)  +\frac{x_{i}}{\theta+x_{i}}$ which
is indeed 1) Lipschitz and 2) it holds $\frac{\partial\tilde{x_{i}}}{\partial
x_{i}}\rightarrow\infty,x_{i}\rightarrow\infty$. Hence, we have from D2) in
Theorem \ref{thm:corollary1} that%
\[
\frac{\partial h_{0}}{\partial\tilde{x_{i}}}\left(  \tilde{\bm{x}}\right)
\geq\log^{1+\epsilon}\left(  x_{i}\right)
\]
is required in each component which is much weaker that Meyn's condition
(\ref{eq:CondLog}) (note: both approaches require Lipschitz-continuity of
$h_{0}$!).
}
\subsection{Reducing Complexity by Pick and Compare}

\label{sec:pick-and-compare}

Usually, MaxWeight-type policies and corresponding algorithms have a high
complexity, and this also applies to a policy based on the $\mu$-MaxWeight
framework. As the number of control vectors grows exponentially in $l$, a
large computational burden arises from (\ref{eqn:muMW}), which is carried out
in every time slot. Popular approaches to tackle the complexity issue include
randomized pick-and-compare based methods \cite{Tassiulas98}, which reduce
complexity at the expense of higher delay, and Greedy/Maximal scheduling
\cite{Joo2009}, which has good delay performance but achieves only a fraction
of the throughput region (i.e., the set of all arrival rate vectors for which
the network is stabilizable).

In this paper, we apply a randomized version of $\mu$-MaxWeight, based on the
pick-and-compare approach, which preserves throughput optimality
\cite{Tassiulas98}\cite{Eryilmaz2007}. Tailored to $\mu$-MaxWeight, the
approach can be summarized as follows:
{We apply} $\bm{u}(0)\in
\mathcal{U}^{\ast}$ arbitrarily. Afterwards, in each timeslot $t>0$, we first
pick a control $\hat{\bm{u}}\in\mathcal{U}^{\ast}$ randomly. Second, we let
$\bm{u}(t)=\hat{\bm{u}}$, if $\langle
\bm{\mu}(\bm{x}),\bm{B\hat{u}}+\bm{\alpha}\rangle< \langle
\bm{\mu}(\bm{x}),\bm{Bu(t-1)}+\bm{\alpha}\rangle$, and otherwise, we choose
$\bm{u}(t)=\bm{u}(t-1)$. This algorithm is throughput optimality as long as
$\Pr\left( \bm{\hat{u}} = \bm{u}^{\ast}\right)  \geq\delta$, for $\delta>0$
\cite{Tassiulas98} (which is trivially satisfied). The reduced complexity,
however, comes at the expense of a higher convergence time. Yet, a tradeoff
can be achieved by repeatedly applying the pick and compare steps in every
particular timeslot.

Note that the randomized algorithm was used as a basis for decentralized
throughout optimal control policies \cite{Eryilmaz2007}\cite{Eryilmaz10}.

\section{Cost-Function Choice}

\label{sec:costfunctions}

A vital design choice {in the $\mu$-MaxWeight framework} is
the underlying cost-function. Different applications have different
requirements and consequently require different cost functions.
{Subsequently,} we will consider both minimum and maximum
buffer state constraints.

A straightforward choice is the $l_{1}$-norm, i.e., a \textit{linear cost
function} of the form
\begin{equation}
c(\bm{Q})=\sum_{i} c_{i}Q_{i}.\label{eq:CF-l}%
\end{equation}
It is used to minimize the total buffer occupancy,
which is corresponding to {the overall} end-to-end delay.
However, this cost function is unsuitable to avoid buffer underflows, since it
does not penalize buffer states below the target level.

Assume we want to find a cost function that steers the particular buffer
levels towards a target buffer state $\tilde{Q}$. A simple cost function
choice that penalizes deviations from target buffer state $\tilde{Q}$ in both
directions is the \textit{shifted quadratic cost function}
\begin{equation}
c(\bm{Q})=\sum_{i} c_{i}(Q_{i}-\tilde{Q})^{2}.\label{eq:CF-q}%
\end{equation}
However, (\ref{eq:CF-q}) naively treats all buffers in the network equally,
although, most likely, only application buffers have minimum state constraints.

Therefore, we combine \eqref{eq:CF-l} and \eqref{eq:CF-q}, such that only
application buffers have quadratic cost terms. Thereby, we obtain the
\textit{composite cost function}, given by
\begin{equation}
c(\bm{Q})=\sum_{i\in\mathcal{I}_{u}} c_{i}(Q_{i}-\tilde{Q})^{2} +
\sum_{j\notin\mathcal{I}_{u}} c_{j}Q_{j},\label{eq:CF-m}%
\end{equation}
where $\tilde{Q}$ denotes the desired target level for application buffers,
and $\mathcal{I}_{u}$ is the set of application buffer indices. For
simplicity, we assume all application buffers have the same target buffer level.

\subsection{Policy Design}

\label{sec:policydesign}

Having chosen an appropriate cost function, we now show how to construct a
corresponding weight function $\bm{\mu}$ in \eqref{eqn:muMW}. In order to
guarantee stability, {the} weight function needs to fulfill
the stability conditions of Theorem \ref{thm:theorem1}. For this purpose, we
employ a perturbation technique (cf. \cite{Meyn:2007}). In \cite{Wunder:2012},
we pointed out that a simple way to construct a weight function is
\begin{equation}
\bm{\mu}(\bm{x})=\mathbf{P}_{\theta}\left(  \mathbf{x}\right)  \nabla
h_{0}(\bm{x}),\label{eq:weight_function_construction}%
\end{equation}
where a \textit{perturbation matrix} $\mathbf{P}_{\theta}\left(
\mathbf{x}\right)  :=\mathbf{diag}\left(  1-\exp{\left(  -\frac{x_{i}}%
{\theta(1+\sum_{j\neq i}x_{j})}\right)  }\right)  $ is used. It is based on
the following perturbation of variables:
\begin{equation}
\tilde{x}_{i}:=x_{i}+\exp{\left(  -\frac{x_{i}}{\theta(1+\sum_{j\neq i}x_{j}%
)}\right)  }.\label{eq:pert2}%
\end{equation}
\inma{
It can be easily verified that the conditions from Theorem
\ref{thm:corollary1} hold for our cost function choice (\ref{eq:CF-m}) which
results in the policy%
\begin{equation}
\mu_{i}(\bm{x})=2c_{i}\left(  {x_{i}}-\bar{q}\right)  \left(  1-\exp{\left(
-\frac{x_{i}}{\theta(1+\sum_{j\neq i}x_{j})}\right)  }\right)
.\label{eq:cf_m_policy}%
\end{equation}
The gradient is given by $g(\bm{x}):=\exp({-x_{i}/}\theta(1+\sum_{j\neq
i}x_{j}))$%
\[
\left[  \nabla\mu_{i}(\bm{x})\right]  _{k}=\left\{
\begin{array}
[c]{cc}%
2c_{i}g(\bm{x})+\frac{2c_{i}\left(  x_{i}-q\right)  g(\bm{x})}{\theta
(1+\sum_{j\neq i}x_{j})} & k=i\\
\frac{-2c_{i}\left(  x_{i}-q\right)  x_{i}g(\bm{x})}{\theta(1+\sum_{j\neq
i}x_{j})^{2}} & else
\end{array}
\right.
\]
and we can see that it contains functions like $ye^{-y}$ or $y^{2}e^{-y},y>0,$
where $y(\bm{x}):={x_{i}/}\theta(1+\sum_{j\neq i}x_{j})$ which at best just
cannot decrease for any path $y(\bm{x})$ in the $m$-dimensional positive
orthant. On the other hand, we have from (\ref{eq:cf_m_policy}) that
$\left\Vert \boldsymbol{\mu}(\bm{x})\right\Vert _{1}$ grows at least linear in
any of the $x_{i}$ so that we have proved $\left\Vert \nabla\mu_{i}%
(\bm{x})\right\Vert _{1}\leq\epsilon\left\Vert \boldsymbol{\mu}%
(\bm{x})\right\Vert _{1}$ which is exactly the condition of throughput
optimality in Theorem \ref{thm:corollary1}.
}

As we will demonstrate now, the choice of the underlying cost function is
crucial for the performance of the resulting control policy.
{However, it should be pointed out that, in network design, it is important to keep the practical limitations of a control framework in mind.
For example, although, owing to the cost-funtion based approach, the end-to-end delays can be better controlled than in classical MaxWeight, hard delay guarantees cannot be given. This may limit the application of the control framework for certain delay-critical network applications (e.g., in the context of transport networks with so-called mission-critical communication traffic).}

\subsection{Example: Controlling a network of queues in tandem}

\label{sec:tandem}

Let us first consider a very simple network, known as tandem queue, comprising
a number of $m$ buffers in series.
We assume traffic arrives at the first buffer with mean rate $\alpha$, and the
application removes traffic from the $m$'th buffer at a constant rate $R_{a}$
(this models, e.g., a streaming service). The output of buffers $1$ through
$m-1$ can be regulated by the control policy. While tandem queue networks have
been thoroughly investigated (see, e.g., \cite{MeynBook2007}), we have two
additional difficulties here. First, we have no explicit control over the rate
at which data is extracted from application buffer $Q_{m}$. Second, in
addition to stability, we require also a minimum buffer state at the
application buffer.

Consider now the most simple network of $m=2$. Here, the only remaining
control decision is whether to send traffic from $Q_{1}$ to $Q_{2}$, or not.
Figure \ref{fig:cfs_tandem} shows, how we can steer the buffer state towards a
target buffer level, using an appropriate cost function. It depicts
{the} queueing trajectories
{of the two queues in the tandem network, i.e., the buffer sizes over time,}
(exponentially averaged over a window of 100 time slots to increase
readability) for 20000 time slots using two different cost functions. The
dashed lines represent policies based on the simple linear cost function in
(\ref{eq:CF-l}), which, obviously, does not stop the second buffer from
growing. By contrast, the solid lines represent the composite cost function
(\ref{eq:CF-m}), which stops $Q_{2}$ from sending further traffic to $Q_{2}$
when the latter reaches a certain level. Instead, the excessive traffic is
buffered at $Q_{1}$, since it generates lower costs. \begin{figure}[ptb]
\centering
\includegraphics[width=0.6\linewidth]{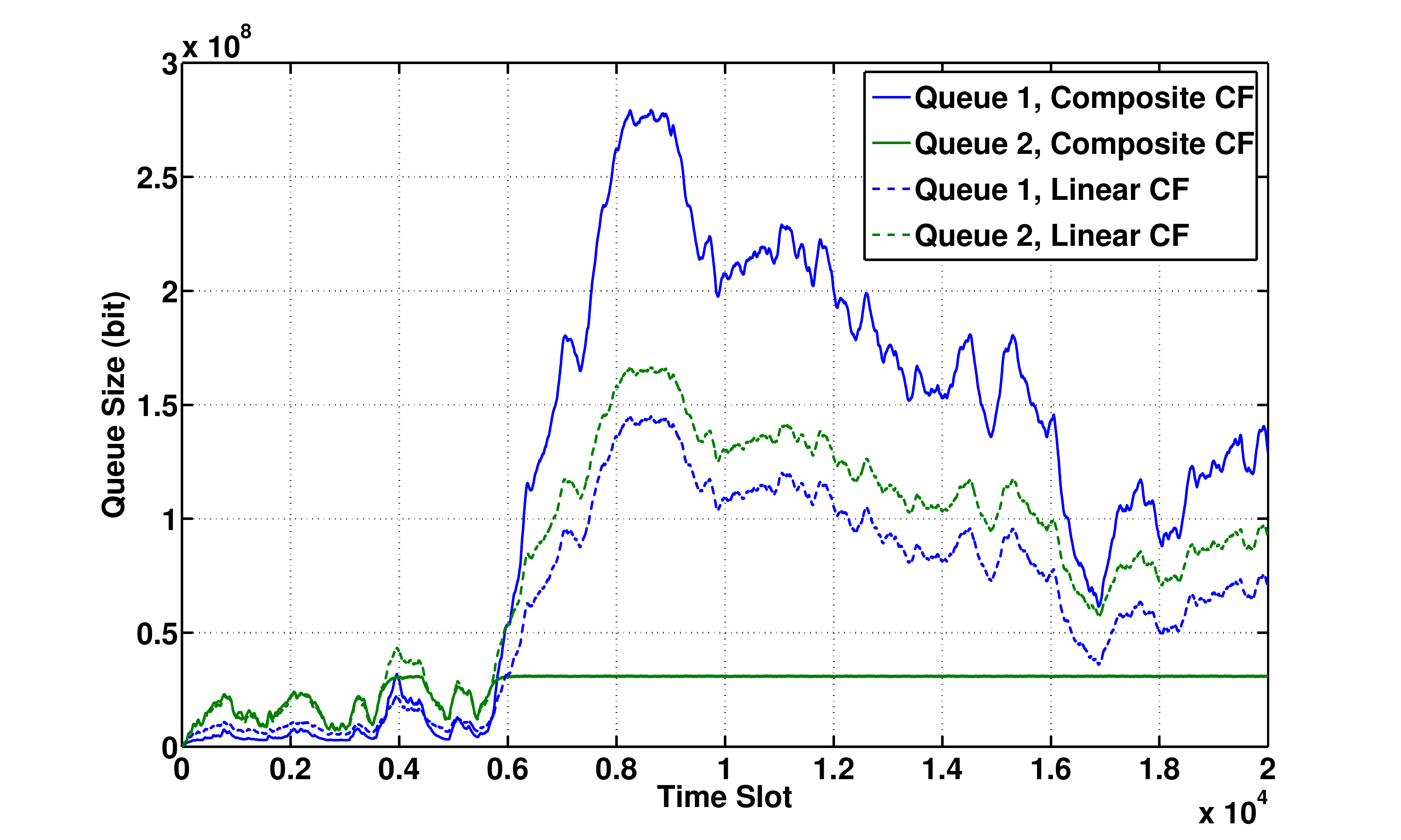}  \caption{Queue
trajectories
{(i.e., number of bits in a particular queue at a particular time slot)}
{obtained with the composite and linear}
cost-functions in {a} tandem queue network
{(i.e., a network of two buffers in series)}.
{The trajectories are thereby exponentially averaged over a time horizon of 100 time slots.}}%
\label{fig:cfs_tandem}%
\end{figure}

\section{Applications}

\label{sec:applications}

In the following, we discuss several examplary fields of application for the
control framework:

(1) the control of a multimedia network carrying streaming traffic,  (2) the
design of an energy efficient policy for energy constrained networks, and  (3)
the cross-layer control of wireless networks for interference management.

\subsection{Application I: Control of multimedia networks}

\label{sec:simulations}

As a first practical application for our framework, we  evaluate the $\mu
$-MaxWeight approach in a network designed for entertainment purposes.
Therefore, the network is expected to carry mostly streaming traffic (and
can, for example, be used to model a wireless entertainment system in  an
aircraft cabin, as in \cite{KasparickFNMS12}). Figure
\ref{fig:syst_mod_schematic} depicts the considered network schematically.
\begin{figure}[ptb]
\centering
\includegraphics[width=0.6\linewidth]{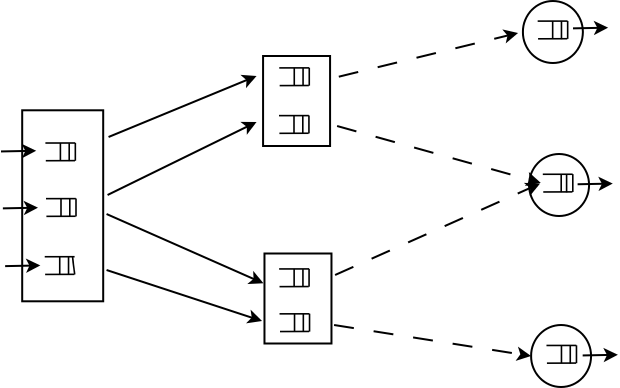}
\caption{Schematic representation of considered multimedia
network.
{A central application server is depicted (left) with three queues, corresponding to the three user terminals (right). In between, we have two access points, where each can serve two of the three users. Dashed lines indicate wireless connections, solid lines indicate fixed wired connections. }}%
\label{fig:syst_mod_schematic}%
\end{figure}We have several wireless access points
{(depicted as smaller rectangles in Figure \ref{fig:syst_mod_schematic})},
each serving a certain number of user terminals in its area
{(depicted as circles in Figure \ref{fig:syst_mod_schematic})}.
The user terminals are assumed to run streaming-based applications. Some
terminals, which are located in between several access points, can potentially
be served by more than one access point. The access points are connected by a
wired backbone network to a central application server
{(depicted as large rectangle in Figure \ref{fig:syst_mod_schematic})}.
The server itself is connected to the Internet, thus, traffic for each user
arrives in a random fashion. Each component in the system has a number of
queues with different requirements.

To account for the anticipated multimedia applications, we define the notion
of \emph{queue outage} as a measure of a control policy's performance (in
addition to the average cost performance metric). We denote the total number
of buffer underflow events and buffer overflow events up to time $T$ as
$F^{\text{min}}_{i}(T)$ and the $F^{\text{max}}_{i}(T)$, respectively. The
total sum of buffer outages is consequently $\bar{F}^{\text{out}}(T) =
\sum_{i\in\mathcal{I}_{u}}\left( F^{\text{min}}_{i}(T) + F^{\text{max}}%
_{i}(T)\right) $, and the relative frequency of queue outage events is defined
as
\begin{equation}
\bar{P}^{\text{out}}(T) = \frac{1}{T\cdot|\mathcal{I}_{u}|} \bar
{F}^{\text{out}}(T).\label{eq:queueOutEv}%
\end{equation}
Assume, we want to keep the buffer states between a minimum buffer state
$Q^{(1)}$ and a maximum buffer state $Q^{(2)}$, which can be flexibly adopted
to the requirements of the desired application. Then, a reasonable choice for
the target buffer state is $\tilde{Q} = \frac{1}{2}\left( Q^{(1)}%
+Q^{(2)}\right) $.

Usually, queueing network models, such as the CRW model, assume static links.
However, we require wireless links between access points and terminals. For
this, we model the wireless link capacities by applying a result from
\cite{Collings08}, which determines the mutual information distribution of a
multi-antenna {OFDM}-based wireless system. To obtain rate expressions, we use
(similar to \cite[Sec. IV.B]{Collings08}) the notion of outage capacity
$\mathcal{I}_{\text{out},p_{o}}$, which is given by \cite[Eq. (26)]%
{Barriac04}
\[
\mathcal{I}_{\text{out},q} = \mathbb{E}[\mathcal{I}_{\text{OFDM}}]
-\sqrt{\mathrm{Var}[\mathcal{I}_{\text{OFDM}}]}Q^{-1}(p_{o}),
\]
for a given outage probability $p_{o}$, and with $Q(\cdot)$ being the Gaussian
Q-function. This outage probability (not to be confused with the queue outage
in \eqref{eq:queueOutEv}) is defined as the maximum rate that is guaranteed to
be supported for $100(1-p_{o})\%$ of the channel realizations. $\mathbb{E}%
[\mathcal{I}_{\text{OFDM}}]$ and $\mathrm{Var}[\mathcal{I}_{\text{OFDM}}]$ are
determined according to \cite{Collings08}.

Subsequently, we evaluate the earlier defined cost functions with respect to
the queue outage performance, and with varying traffic intensities, in order
to assess the robustness of the different cost functions. In particular, we
show results (after $T=100000$ time slots) in a network according to Figure
\ref{fig:syst_mod_schematic}, with 3 access points and 10 users per access
point. We assume the capacity of the wired links is $100\,\mathrm{Mbit/s}$,
the target buffer size is $\tilde{Q}=20\,\mathrm{Mbit/s}$, and the rate at
which the application drains the user buffers is $3\,\mathrm{Mbit/s}$.
Moreover, we allow a wireless link outage probability of $p_{o}=0.01$.
Moreover, we apply $100$ pick-and-compare iterations per timeslot (cf. Section
\ref{sec:pick-and-compare}).

First, we consider the buffer underflow probability, since our main motivation
is to prevent service interruptions due to low buffers. Since the application
drains the application queues at a fixed rate, intuitively, one can expect
that at traffic lower than this value, the influence of underflows dominates,
while at traffic rates larger than this value overflows are more likely to
occur. Consider Figure \ref{fig:compCostFun3}, which compares the underflow
probability of various cost functions, and, as a baseline, that of MaxWeight.
When the mean arrival rate equals the application service rate $R_{a}$ (or is
larger), all policies produce low underflow frequencies, however, we already
observe small gains from cost functions based approach. The gain significantly
grows when the arrival rates are below $R_{a}$. While MaxWeight and the linear
{CF} \eqref{eq:CF-l} show almost the same high underflow frequency (since both
policies do not penalize low buffer states at the application buffers),
already the shifted quadratic {CF} \eqref{eq:CF-q} shows improvements. The
best performance is obtained with the composite CF in \eqref{eq:CF-m}, which
assign different costs to application buffers and non-application buffers.

However, not only buffer underflows have to be avoided, but it may be
desirable to avoid large queues as well. Therefore, we subsequently
investigate the performance with respect to the queue outage probability
defined in \eqref{eq:queueOutEv}. Figure \ref{fig:compCostFun2} summarizes the
results. \begin{figure}[ptb]
\centering
\subfigure[Relative frequencies of buffer underflow events.]  {
\includegraphics[width=0.48\linewidth]{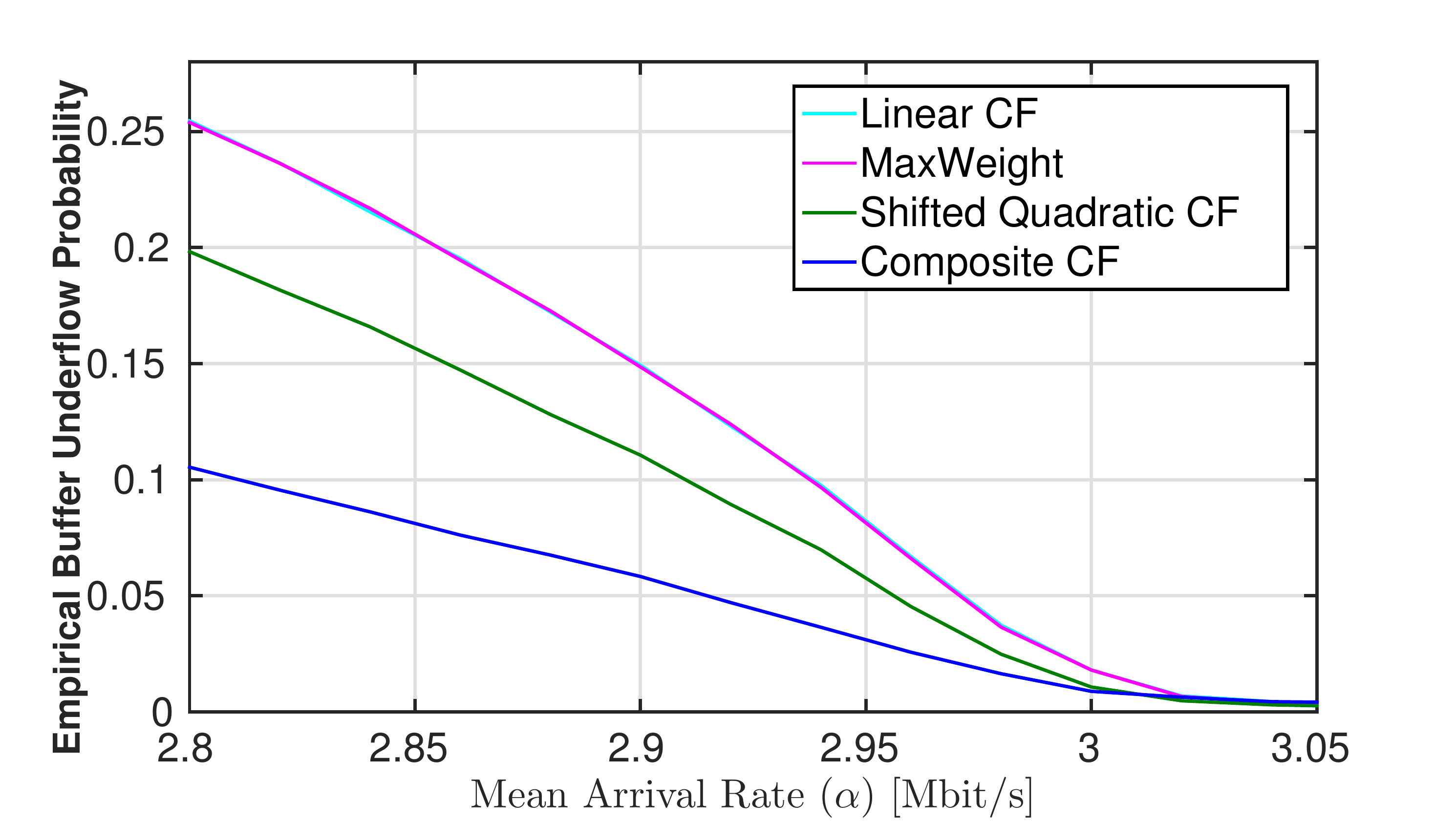}
\label{fig:compCostFun3}  }
\subfigure[Relative frequencies of queue outage events.]  {
\includegraphics[width=0.48\linewidth]{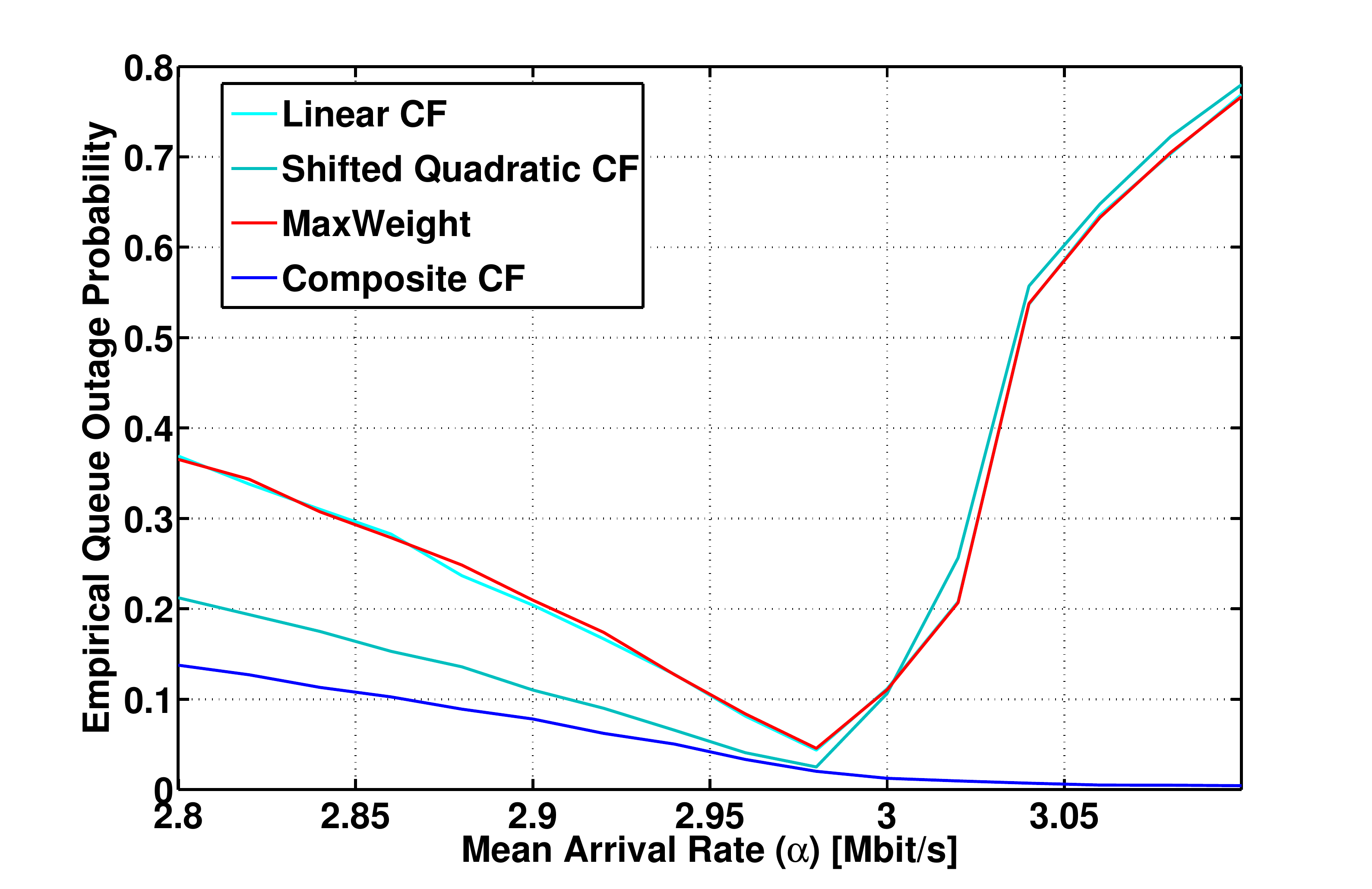}
\label{fig:compCostFun2}  }\caption{Performance results after 100000 time
slots obtained by different cost functions
{with respect to buffer underflows (left) and queue outages (right). The target buffer size is $\tilde{Q}=20\,\mathrm{Mbit/s}$,
and the application rate $R_a = 3\,\mathrm{Mbit/s}$. }}%
\label{fig:perf_res_cf}%
\end{figure} As long as the arrival rates are below $R_{a}$, all cost
functions produce a decreasing outage frequency when the arrival rates
increase, mainly since queue underflows are less frequent. Beyond the
applications service rate, the policies that are not based on a sophisticated
cost function rapidly increase the queue outages with increasing traffic,
owing to a higher number of overflows. Only the composite cost function can
further decrease the queue outages, since exceeding traffic is stored at
buffers that generate lower costs.

\subsection{Application II: Energy Efficient Network Control}

\label{sec:EnergyEfficiency}

As another application, we subsequently demonstrate how the $\mu$-MaxWeight
framework can be used for energy efficient (``green'') routing. In particular,
wireless ad hoc networks, such as sensor networks, often comprise energy
limited, battery operated nodes. In energy constrained networks, routing
policies may {require significant modifications}, such that they
extend the lifetime of nodes.
{Although a considerable amount of work}
 has been {conducted}
 on energy efficient
routing, only few {studies}
consider the combination of energy efficiency and throughput optimality, e.g.,
\cite{Neely_Energy_it}.

Consider the exemplary system depicted in Figure \ref{fig:systemEnEff}. We
have a source of data, a corresponding destination, and three different routes
in between (a similar topology was investigated, e.g., in \cite{Zhu2011}).
\begin{figure}[ptb]
\centering
\includegraphics[width=0.6\linewidth ]{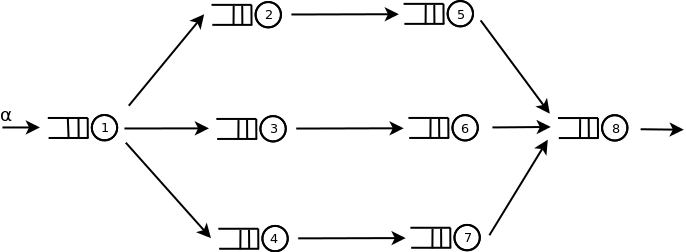}
\caption{Network of battery powered nodes
{with unidirectional links}.
{Traffic enters the network at node 1, and leaves the network at node 8.}
{A node can switch to an energy saving mode, if its queue is empty. Therefore, to be}
energy efficient, {the} control policy should route the
traffic from source to destination {over}
 as few {nodes} as possible.}%
\label{fig:systemEnEff}%
\end{figure}{Let us} assume that the intermediate notes are
battery powered devices. Moreover, {we} assume that a node can
switch to an energy saving mode whenever there is no traffic in its queues. In
this case, an energy efficient routing algorithm tries to maximize the overall
total idle time. However, as in the previous application, stability shall be
guaranteed nevertheless. Naturally, when the network load is low, a suitable
control policy should route as much traffic as possible using only a single
route. This allows the devices that constitute the other routes to be switched
off. Only if the network load is high, additional routes should be used to
avoid queue instability.

This required behaviour is certainly not achieved when {the}
MaxWeight {policy} is used, which, by contrast,
 uses all routes in an attempt to balance buffer
states. Therefore, we apply the $\mu$-MaxWeight framework with a suitable cost
function in order to obtain the desired behaviour at low traffic load. A
straight forward cost function choice in this setup is the linear cost
function \eqref{eq:CF-l}. Let $c_{i}$ be the cost factor associated with the
buffer of node $i$. In order to force the traffic towards the central route,
we should choose our cost weights as  $c_{j} >> c_{i}$, for $j\in\{2,4,5,7\},
i\in\{3,6\}$.

We simulate the system of Figure \ref{fig:systemEnEff} with the following
parameters. The arrival rate $\alpha$ to the queue associated with node 1 is
varied in the interval of $\alpha\in[0.1,0.5]$. All other queues have no
external arrivals. As a benchmark, we use the MaxWeight policy again.
 The results {are} depicted in
Figure \ref{fig:eneff-results}. The figure shows the total aggregated idle
time, defined as
\begin{equation}
\sum_{t=1}^{T} \sum_{m} \mathbb{I}\{q_{m}(t)=0\}\label{eq:SumIdle}%
\end{equation}
over the load $\alpha$. Thereby, we consider only the energy-limited
intermediate nodes, that is, $m\in[2,7]$ in \eqref{eq:SumIdle}. In Figure
\ref{fig:eneff-results-eq}, we set the capacities of all links to $0.5$, while
in Figure \ref{fig:eneff-results-neq} we set the capacities of the central
route higher than those of the upper and lower route.

\begin{figure}[ptb]
\centering
\subfigure[Equal Link Capacities]{\label{fig:eneff-results-eq}\includegraphics[width=0.48\linewidth ]{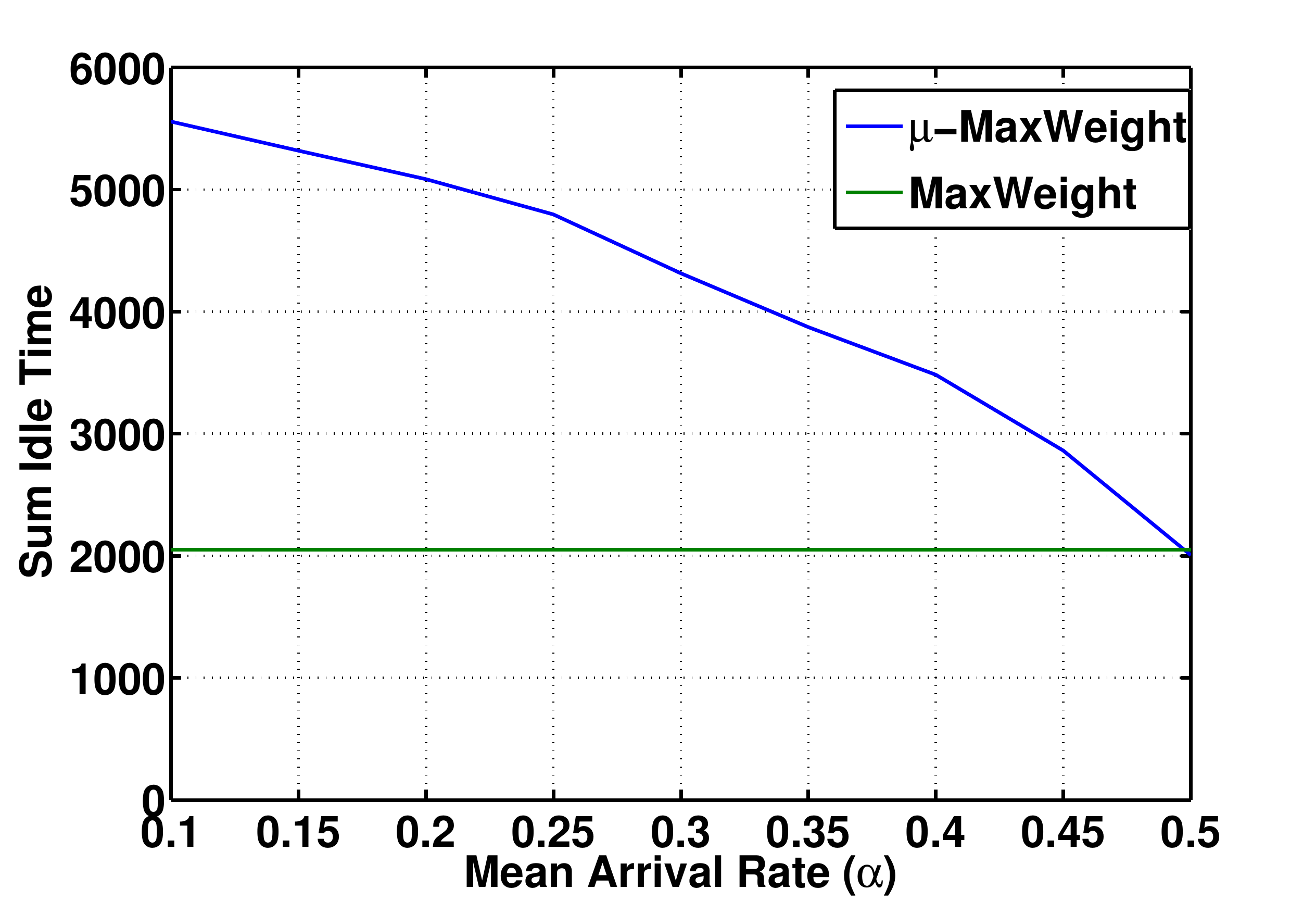}}
\subfigure[Unequal Link Capacities]{\label{fig:eneff-results-neq}\includegraphics[width=0.48\linewidth ]{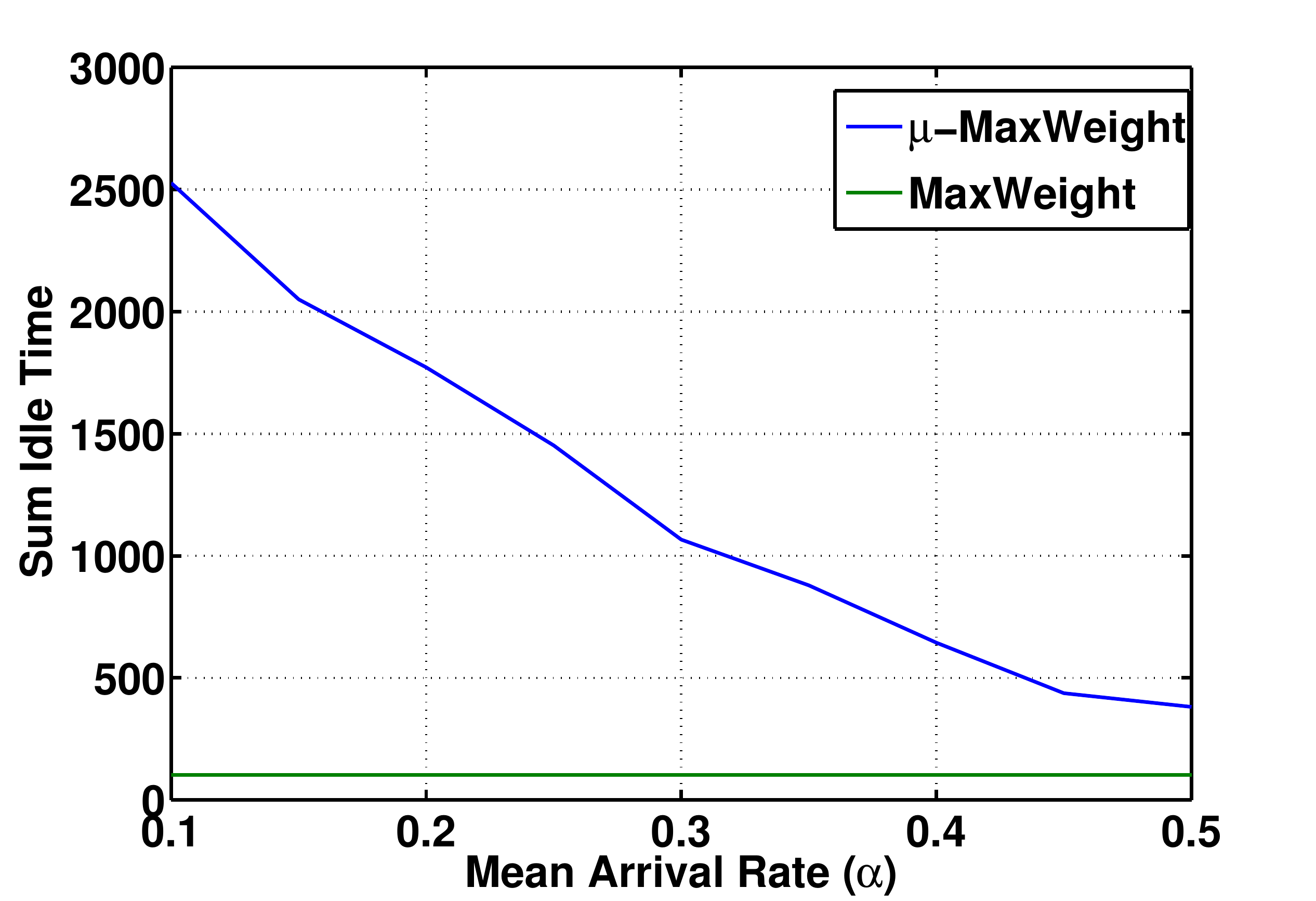}}
\caption{Sum idle time
{(aggregated over all relevant queues)} as a measure of energy
efficiency, {plotted} over
{the average arrival rate $\alpha$}.
{The performance of the $\mu$-MaxWeight policy with a linear cost function is compared to the MaxWeight baseline algorithm.}
}%
\label{fig:eneff-results}%
\end{figure}

Clearly, the proposed control framework outperforms MaxWeight in terms of
energy efficiency. Even in case of equal routes, where the relative gains are
smaller, the performance is still significantly better. Note that as the
traffic load increases, the behaviour of $\mu$-MaxWeight gets closer to that
of MaxWeight, which is necessary to ensure stability.

\subsection{Application III: Cost-Based Cross-Layer Control}

\label{sec:WSA}

As a last example, we show how the proposed network control framwork can be
combined with physical layer resource allocation in a common cross-layer
control problem. While in network models  with
fixed link capacities the coupling of link rates due to interference is
usually not considered, we {here} assume a fully coupled
wireless network, where nodes treat interference as noise. To cope with the
non-convex power control problem, we apply the SCA method from
\cite{Papandriopoulos2009}. In \cite{Matskani2012_TSP}, cross-layer control
based on MaxWeight was already successfully combined with the SCA technique.

Consider a wireless multihop network comprising a set of nodes $\mathcal{N}$.
Traffic can potentially be generated at all nodes in the network and is
categorized by its destination. That is, traffic for node $c$ is called
commodity-$c$ traffic. Nodes are connected by links $l\in\mathcal{L}$, which
are defined as ordered node tuples $(i,j)$, with $i,j\in\mathcal{N}$.
Furthermore, we denote $\mathcal{O}(j)$ the set of all outgoing links of node
$j$. The set of all incoming links is accordingly defined as $\mathcal{I}(l)$.
We define $o(l)$ as the transmitter of link $l$ and $d(l)$ as the receiver of
link $l$. Let $g_{lm}$ be the path gain from node $o(l)$ to node $d(m)$. Thus,
the gain of link $l$ is denoted $g_{ll}$. Accordingly, let $d_{ij}$ be the
distance between node $i$ and node $j$. A frequently used simplification for
static gains is that, depending on their distance, the link between two nodes
can be modeled as:  $g_{ij} = \frac{1}{\left( d_{ij}\right) ^{\rho}}$, with
$\rho\in\mathbb{R}$ being the path-loss exponent. Below, we assume a path-loss
exponent $\rho=4$, which corresponds to a rather lossy environment (e.g., an
urban or sub-urban environment).

Each node $j$ assigns a power value $p_{l}$ to each of its outgoing links
$l\in\mathcal{O}(l)$. Let $\bm{p}\in\mathbb{R}^{L}$ be the vector comprising
the current power allocation of all $L$ links in the network. Here, we
consider per-node power constraints, where all nodes have a common maximum
available power value of $P_{\text{max}}$, thus  $\sum_{l\in\mathcal{O}%
(l)}p_{l} \leq P_{\text{max}}$. However, all algorithms can be easily applied
in case of per-link power constraints as well. Let
\begin{equation}
\gamma_{l}(\bm{p})=\frac{g_{ll}p_{l}}{\sigma_{l}^{2}+\sum_{j\neq l}g_{jl}%
p_{j}}\label{eq:SINR}%
\end{equation}
be the {SINR} experienced at node $d(l)$, with $\sigma_{l}^{2}\in\mathbb{R}$
being the noise power. Consequently, the rate of link $l$, depending on the
current power allocation $\bm{p}$, is given by
\begin{equation}
r_{l}(\bm{p}) = W\cdot\log_{2}\left( 1+\gamma_{l}(\bm{p})\right)
,\label{eq:linkRates}%
\end{equation}
where $W\in\mathbb{R}$ is the available system bandwidth.

Our approach is based on the well-known {BPPC} algorithm (see
\cite{Georgiadis_06}\cite{Matskani2012_TSP} for details on this algorithm),
which is known to be throughput optimal \cite{Georgiadis_06}.

Subsequently, we show how to combine the {BPPC} approach with $\mu$-MaxWeight.
The main difference to the procedure aboveis the generation of link weights.
{As before, and} in contrast to a MaxWeight-based algorithm
that uses only the buffer states to determine the weights, we determine the
weights using our weight functions $\bm{\mu}$. Using, for example, a linear
cost function, this results in  $\left[ \bm{\mu}(\bm{x})\right] _{i} =
c_{i}\left( 1-\exp\left( -\frac{x_{i}}{\theta(1+\sum_{j\neq i}x_{j})}\right)
\right) $, such that
\begin{equation}
w_{l}^{\ast}(t) = \left( \left[ \bm{\mu}(\bm{x})\right] _{o(l)}^{(c)}-\left[
\bm{\mu}(\bm{x})\right] _{d(l)}^{(c)}\right) u^{\ast}_{l}%
.\label{eq:muMWLinkWeights}%
\end{equation}
The optimization problem at the physical layer is now
\begin{align}
\underset{\bm{p}}{\text{max.}} \quad &  \sum_{l}w_{l}^{\ast}(t)\cdot\log\left(
1+\gamma_{l}(\bm{p})\right)  \label{eq:OptProbNonConvex}%
\end{align}
subject to $p_{l}\leq P^{\text{max}}$ ($\forall l$), and
{with} weights defined in \eqref{eq:muMWLinkWeights}. Using
rates defined in \eqref{eq:linkRates}, we still have a difficult non-convex
problem, making a global optimization prohibitively complex in large networks.

Therefore, we apply the {SCA} algorithm of \cite{Papandriopoulos2009}%
\cite{Matskani2012_TSP}, which can be shown to converge at least to a {KKT}
point of \eqref{eq:OptProbNonConvex} \cite{Papandriopoulos2009}. Core of the
{SCA} algorithm is to convexify \eqref{eq:OptProbNonConvex} by iteratively
applying an appropriate lower bound, given by
\begin{equation}
\alpha_{l}\log(\gamma_{l}(\bm{p}))+\beta_{l} \leq\log\left( 1+\gamma
_{l}(\bm{p})\right) ,\label{eq:lowerBound}%
\end{equation}
to the link capacity, followed by a logarithmic change of variables, i.e.,
$\tilde{p}_{l} := \log(p_{l})$. Parameters $\alpha_{l}\in\mathbb{R}$ and
$\beta_{l}\in\mathbb{R}$ are chosen according to $\alpha_{l} = \frac{z_{0}%
}{1+z_{0}}$ and $\beta_{l} = \log(1+z_{0}) - \frac{z_{0}}{1+z_{0}} \log
(z_{0})$, with some $z_{0}\in\mathbb{R}$. Note that the bound
\eqref{eq:lowerBound} is tight for $z_{0} =\gamma_{l}(\bm{p})$. The algorithm
is initialized according to $\bm{\alpha}(0)=\bm{1}$ and $\bm{\beta}(0)=\bm{0}$%
, which is equivalent to the high-{SINR} approximation. Each iteration of the
algorithm first includes a maximize-step
\begin{equation}
\tilde{\bm{p}}^{\ast}(t) \in\arg\max_{\bm{\tilde{p}}}\sum_{l}w_{l}^{\ast
}(t)\tilde{r}(\tilde{\bm{p}}), \label{eq:CONVEXSUBPROBLEM}%
\end{equation}
subject to $\exp(\tilde{p}_{l}) \leq\bar{p}_{l}$, whose solution is then used
in a tighten-step to update $\alpha_{l}$ and $\beta_{l}$ for each link
according to   $\alpha_{l}(t+1)=\frac{\gamma_{l}(\bm{p}^{\ast}(t))}%
{1+\gamma_{l}(\bm{p}^{\ast}(t))}$  and $\beta_{l}(t+1)=\log\left( 1+\gamma
_{l}\left( \bm{p}^{\ast}\left( t\right) \right) \right) -\alpha_{l}\log\left(
\gamma_{l}\left( \bm{p}^{\ast}\left( t\right) \right) \right) $.  When
$\hat{r}_{l}\left( \bm{p}^{\ast}(t),\bm{\alpha}(t), \bm{\beta}(t)\right)
\approx\hat{r}_{l}\left( \bm{p}^{\ast}(t),\bm{\alpha}(t+1),
\bm{\beta}(t+1)\right) $, the algorithm terminates.

To evaluate the performance, we conduct numerical simulations over 10000 time
slots and compare to various baselines. A main advantage of the cost function
based approach is that every buffer can be weighted with different
coefficients $c_{i}$. Doing so, the use of specific buffers (or entire routes)
can be discouraged, for example, in cases where buffer space is more expensive
at selected nodes than elsewhere. We perform simulations both with unequal and
equal weights for all buffers in a network with $N=9$ nodes connected by
$L=10$ links. Furthermore, we assume a single source node and 5 possible
destination nodes, with the same average arrival rate for all traffic streams.
{The network has a layered topology that is larger, but structurally similar to the network that is depicted in Figure \ref{fig:syst_mod_schematic}, however with 3 nodes in the second layer and with 5 nodes in the third layer.}
A bandwidth of 20\,MHz is available for the wireless links.

In Figure \ref{fig:avgQueues}, we compare the average sizes of all queues in
the network over the first $200$ time slots, both with (left) and without
(right) power control. In the latter case, each node simply distributes the
available power equally over the currently active outgoing links. In case
{SCA}-based power control is used, all queues converge to a steady mean, while
in the no-power-control case some queues grow without bound. Therefore, we can
conclude that the stability region is larger due to the convex approximation
based power control algorithm. \begin{figure}[ptb]
\centering
\includegraphics[width=0.6\linewidth ]{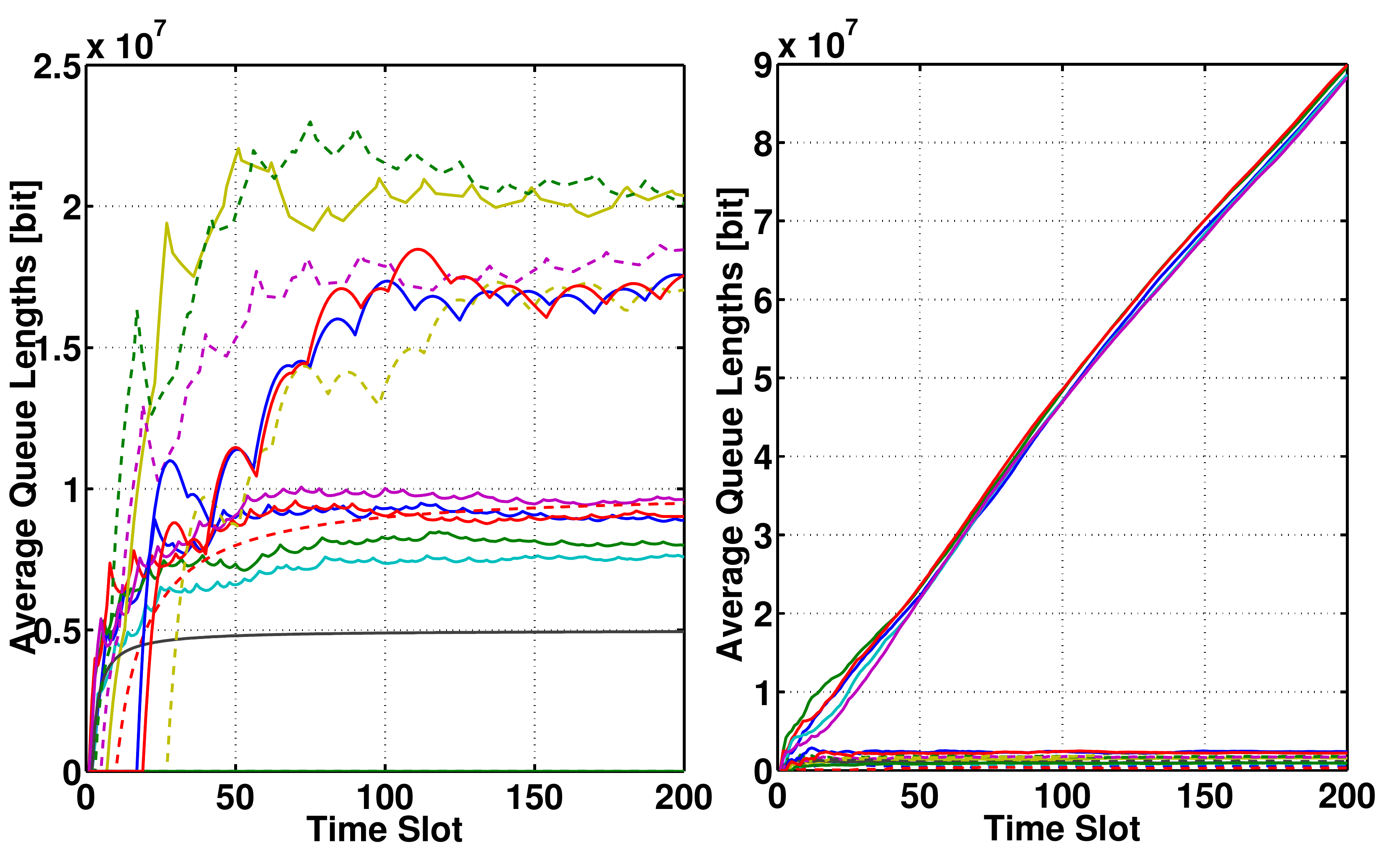}
\caption{Average buffer states over time obtained by $\mu
$-MaxWeight using {SCA}-based {power control} (left) and no
power control (right).
{The figure shows the trajectories (buffer size over time) of all ($\lvert\bm{Q}\rvert=17$)) queues in the network.}}%
\label{fig:avgQueues}%
\end{figure}In Figure \ref{fig:eqCompPols}, we compare the time averaged cost
incurred by the $\mu$-MaxWeight and {SCA} based policy with that of the
corresponding MaxWeight-based policy (with and without power control based on
{SCA}). In addition, we show the performance of the cross-layer algorithm when
relying on the high-{SINR} assumption in the optimization. Similar to what we
saw before, the stability region without power control is much smaller than
that based on the {SCA} algorithm. Although using the high-{SINR} assumption
in the algorithm can stabilize the network at higher arrival rates (unlike no
power control),  the performance is still considerably worse than that of the
{SCA}-based algorithms. To allow a comparison of the {SCA}-based cross-layer
algorithms using $\mu$-MaxWeight with that using classical MaxWeight (whose
curves are very close in the top half of Figure \ref{fig:eqCompPols}), in the
lower half of Figure \ref{fig:eqCompPols}, we show the average cost for a
relatively high arrival rate over time. After the network reaches a steady
state, significant gains can be observed compared to MaxWeight,
{and these gains}  are larger
when the mean arrival rate is higher. The more traffic traverses the network,
the higher will be the interference between active links, and the more
beneficial becomes power control. \begin{figure}[ptb]
\centering
\includegraphics[width=0.6\linewidth ]{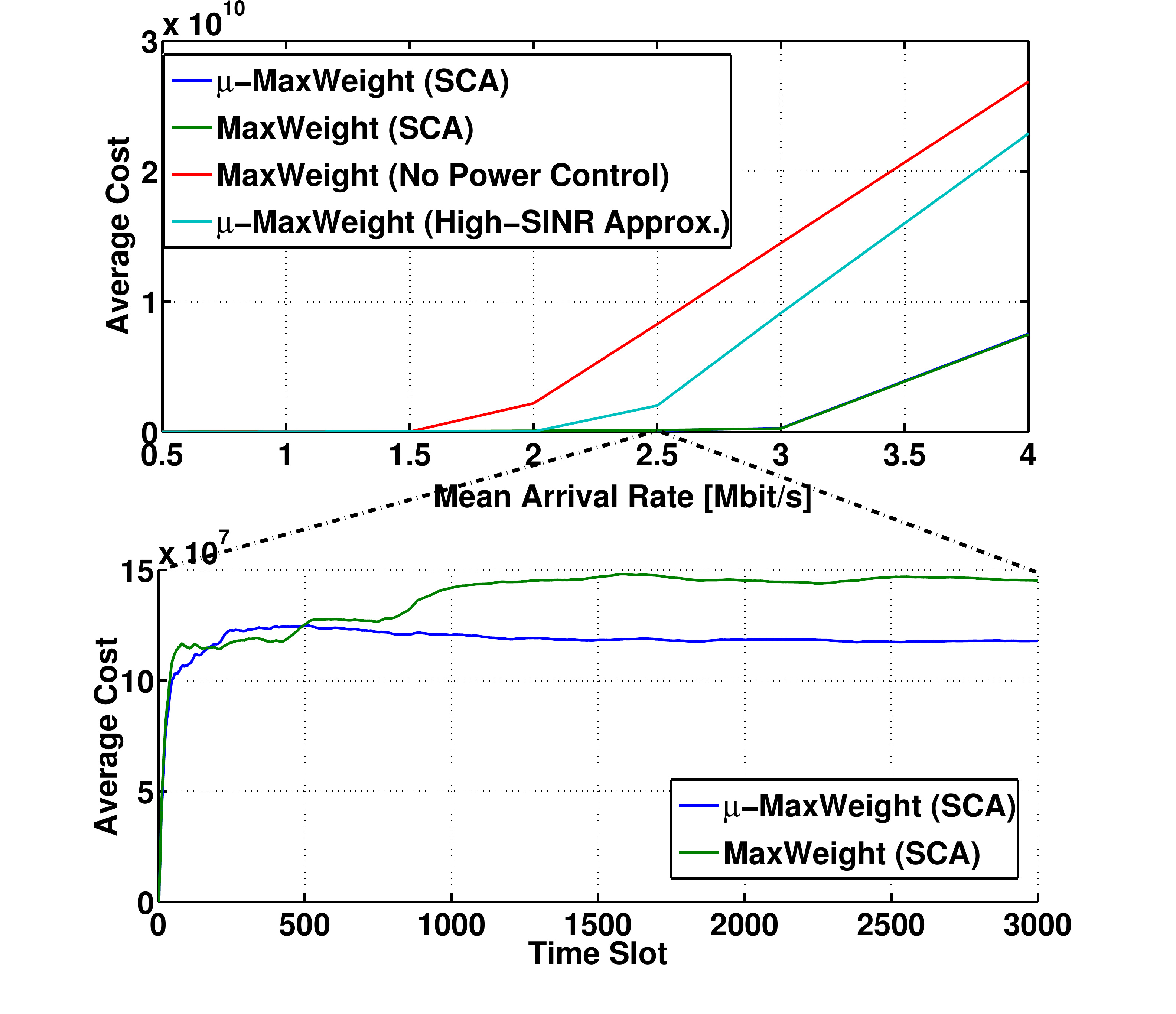}
\caption{Performance comparison of different
{control} policies.
{The $\mu$-MaxWeight based policies are based on a cost function with}
equal weights. The top figure depicts the
average cost over the mean arrival rate
{after a simulation time of 10000 time slots}, while the
bottom figure depicts the average cost over time for a selected average
arrival rate of $\alpha=2.5\,\mathrm{Mbit/s}$.}%
\label{fig:eqCompPols}%
\end{figure}

\section{Conclusions}

\label{sec:conclusion} We presented and evaluated a framework for the design
and control of queueing networks, which combines throughput optimality with
the minimization of an underlying cost function, thereby taking into account
potential service-dependent constraints. We demonstrated how this framework
can be used to flexibly adapt the network control to different applications
with corresponding requirements. Moreover, we adapted a well-known cross-layer
control algorithm to our control framework, and demonstrated how this
non-convex optimization problem can be efficiently approximated using SCA. In
all evaluated applications, numerical evaluations indicate high gains of the
proposed cost function based approach over classical throughput optimal
control policies, such as MaxWeight.

\appendices

\section{Proof of Theorem \ref{thm:corollary1}} \label{sec:Appx1}
By Condition 2) of Theorem \ref{thm:corollary1}, we can assume that the random
walk evolves on $\mathbb{R}_{+}^{m}$. Hence, we can skip Condition 2) of
Theorem \ref{thm:theorem1}, since this condition (as its counterpart in
Theorem \ref{thm:corollary1}) ensures positivity of the random walk. We need
to show that from%
\begin{equation}
\left\Vert \nabla\log\mu_{i}\left(  \bm{x}\right)  \right\Vert \leq
\epsilon,\quad\forall i\in\mathcal{M},\ \Vert\bm{x}\Vert>C_{6}\left(
\epsilon\right)  ,\label{eq:prfCond1}%
\end{equation}
(where $C_{6}\left(  \epsilon\right)  $ is sufficiently large) follows:
\begin{equation}
\left\vert \frac{\mu_{i}(\bm{x}+\Delta\bm{x})}{\sum_{j\in\mathcal{M}}\mu
_{j}(\bm{x}+\Delta\bm{x})}-\frac{\mu_{i}(\bm{x})}{\sum_{j\in\mathcal{M}}%
\mu_{j}(\bm{x})}\right\vert \leq\epsilon.\label{eq:prfCond2}%
\end{equation}
For orientation, let us assume more restrictive conditions first: take
$\mu_{i},$\ $\forall i\in\mathcal{M}$, Lipschitz continuous, and let
$\sum_{j\in\mathcal{M}}\mu_{j}(\bm{x})\rightarrow\infty$, if $\Vert
\bm{x}\Vert\rightarrow\infty$. Note that these conditions already encompasses
Meyn's perturbation (\ref{eq:pertLog}) together with, e.g., a linear cost function.

It is easy to prove the theorem with these assumptions: by the mean value
theorem, we have
$\mu_{i}\left(  \bm{x}+\Delta\bm{x}\right)  =\mu_{i}\left(  \overline
{\bm{x}}\right)  +\nabla_{\bm{x}}^{T}\mu_{i}\left(  \widetilde{\bm{x}}\right)
\Delta\overline{\bm{x}}$,
where $\overline{\bm{x}}$ is an (arbitrary) point on the line connecting
$\bm{x}$ and $\bm{x}+\Delta\bm{x}$, whereas $\widetilde{\bm{x}}$ is a point on
the line connecting $\overline{\bm{x}}$ and $\bm{x}+\Delta\overline{\bm{x}}$.
Since the field is Lipschitz, we have $\nabla_{\bm{x}}^{T}\mu_{i}%
(\overline{\bm{x}})\Delta\overline{\bm{x}}\leq C_{7}$ uniformly. Furthermore,
since the policy is non-idling $\sum_{j\in\mathcal{M}}\mu_{j}\left(
\bm{x}+\Delta\bm{x}\right)  \geq C_{8}$ where the normalization constant
$C_{8}$ can be chosen as large as possible without altering the policy (by the
construction of the policy). Moreover, since $\sum_{j\in\mathcal{M}}\mu
_{j}(\bm{x})\rightarrow\infty,\Vert\bm{x}\Vert\rightarrow\infty$, condition
(\ref{eq:prfCond2}) is equivalent to
$\left\vert \mu_{i}\left(  \bm{x}+\Delta\bm{x}\right)  -\mu_{i}\left(
\bm{x}\right)  \right\vert \leq\epsilon\sum_{j\in\mathcal{M}}\mu_{j}\left(
\overline{\bm{x}}\right)  $
and, again, by the mean value theorem:
$\left\vert \nabla^{T}\mu_{i}\left(  \overline{\bm{x}}\right)  \Delta
\bm{x}\right\vert \leq\epsilon\sum_{j\in\mathcal{M}}\mu_{j}\left(
\overline{\bm{x}}\right)  $.
Here, we tacitly assumed that we have selected $\overline{\bm{x}}$
accordingly. Since $\Delta\bm{x}$ is fixed and by the positivity of $\mu_{i}$
it is sufficient that
$\left\Vert \nabla\mu_{i}\left(  \bm{x}\right)  \right\Vert \leq\frac
{\epsilon}{\left\Vert \Delta\bm{x}\right\Vert }\mu_{i}\left(  \bm{x}\right)
$,
which is equivalent to condition (\ref{eq:prfCond1}) with some $\Vert
\bm{x}\Vert>C_{6}\left(  \epsilon^{\prime}\right)  $ ($\epsilon^{\prime}$
slightly smaller).

Let us now prove the general case. \inma{By the derivation above, }condition
(\ref{eq:prfCond2}) can be written as
\[
\frac{1}{\sum_{j\in\mathcal{M}}\mu_{j}(\bm{x})}\nabla^{T}\mu_{i}%
(\bm{x})\Delta\bm{x}=\epsilon_{n},
\]
for some $\bm{x}$ with $\Vert\bm{x}\Vert>C\left(  \epsilon_{n}\right)  $ where
$\epsilon_{n}$ is a zero sequence and $C\left(  \epsilon_{n}\right)  $ is
strictly increasing for any fixed $\Delta\bm{x}\in\mathbb{R}^{m}$. Now, again,
by the mean value theorem
$\left\vert \frac{\mu_{i}(\bm{x}+\Delta\bm{x})}{\sum_{j\in\mathcal{M}}\mu
_{j}(\bar{\bm{x}})+\nabla^{T}\mu_{j}(\tilde{\bm{x}})\Delta\bm{x}}-\frac
{\mu_{i}(\bm{x})}{\sum_{j\in\mathcal{M}}\mu_{j}(\bar{\bm{x}})-\nabla^{T}%
\mu_{j}(\tilde{\underline{\bm{x}}})\Delta\underline{\bm{x}}}\right\vert
\leq\epsilon$,
where we set $\bar{\bm{x}}$ as before and let $\bm{x}+\Delta\underline
{\bm{x}}=\bar{\bm{x}}$ and $\bar{\bm{x}}+\Delta\bar{\bm{x}}=\bm{x}+\Delta
\bm{x}$. $\underline{\tilde{\bm{x}}},\tilde{\bm{x}}$ are points on the lines
connecting $\bm{x}$ with $\bar{\bm{x}}$, and $\bar{\bm{x}}$ with
$\bm{x}+\Delta\bm{x}$, respectively. Note that $\mu_{j}(\bar{\bm{x}})$ is zero
if and only if $\mu_{i}(\bm{x}+\Delta\bm{x})$ and $\mu_{i}(\bm{x})$ are both
zero since otherwise by condition (\ref{eq:prfCond1}) the gradient would be
zero as well. Since in this case the condition is trivially satisfied so that
we exclude it. Hence,
it follows that
\begin{equation*}
  \left\vert \mu_{i}(\bm{x}+\Delta\bm{x})-\mu_{i}(\bm{x})\frac{\sum
_{j\in\mathcal{M}}\mu_{j}(\bar{\bm{x}})(1+\overbrace{\frac{\nabla^{T}\mu
_{j}(\tilde{\bm{x}})\Delta\bar{\bm{x}}}{\mu_{j}(\bar{\bm{x}})}}^{(A)})}%
{\sum_{j\in\mathcal{M}}\mu_{j}(\bar{\bm{x}})(1\underbrace{-\frac{\nabla^{T}%
\mu_{j}(\underline{\tilde{\bm{x}}})\Delta\underline{\bm{x}}}{\mu_{j}%
(\bar{\bm{x}})}}_{(B)})}\right\vert 
  \leq\epsilon\cdot\sum_{j\in\mathcal{M}}\mu_{j}(\bar{\bm{x}})\left(
1+\frac{\nabla^{T}\mu_{j}(\tilde{\bm{x}})\Delta\bar{\bm{x}}}{\mu_{j}%
(\bar{\bm{x}})}\right)  .
\end{equation*}

We can prove that, because of condition (\ref{eq:prfCond1}), (A) and (B) are
zero sequences: suppose $\nabla^{T}\mu_{j}(\tilde{\bm{x}})$ is non-zero (then
we can stop anyway) then by the repeated application of the mean value
theorem, the denominator of, say, (A) can be written as
$\mu_{j}(\bar{\bm{x}})=\mu_{j}(\tilde{\bm{x}})+\nabla^{T}\mu_{j}%
(\bm{x}_{2})\Delta\bm{x}_{2}$.
This process generates sequences in $\mathbb{R}_{+}^{m}$ with $\tilde
{\bm{x}}=\bm{x}_{1},\bm{x}_{2},...$ and $\Delta\bar{\bm{x}}=\Delta\bar
{\bm{x}}_{1}\subset\Delta\bar{\bm{x}}_{2},...$ which are bounded and hence we
can pick subsequences converging to some set of limit points $\bm{x}_{\infty
}^{(k)},k=1,2,...$. Note that we can restrict the number of limit points to at
most two since by definition every limit point is visited arbitrarily often
and infinitely close and by construction of the sequence there is no
possibility of more than two limit points which neither contain the other in
between them. Take these two limit points with corresponding subsequence
$\bm{x}_{n}^{(k)},k=1,2$: by continuous differentiability we have $\mu
_{j}(\bm{x}_{n}^{(k)})\rightarrow\mu_{j}(\bm{x}_{\infty}^{(k)})$ and
$\nabla\mu_{j}(\bm{x}_{n}^{(k)})\rightarrow\nabla\mu_{j}(\bm{x}_{\infty}%
^{(k)}),k=1,2$. It must also hold in the limit that
$\mu_{j}(\bm{x}_{\infty}^{(1)})+\nabla^{T}\mu_{j}(\bm{x}_{\infty}%
^{(2)})(\bm{x}_{\infty}^{(1)}-\bm{x}_{\infty}^{(2)})=\mu_{j}(\bm{x}_{\infty
}^{(2)})$
(and vice versa). Since then%
\[
\frac{\nabla^{T}\mu_{j}(\bm{x}_{\infty}^{\left(  2\right)  })(\bm{x}_{\infty
}^{(1)}-\bm{x}_{\infty}^{(2)})}{\mu_{j}(\bm{x}_{\infty}^{\left(  2\right)
})\mu_{j}(\bm{x}_{\infty}^{(2)})}\leq\epsilon,
\]
(and vice versa) where $\epsilon>0$ is arbitrarily small by condition
(\ref{eq:prfCond1}) we conclude that $\mu_{j}(\bm{x}_{\infty}^{\left(
1\right)  })=\mu_{j}(\bm{x}_{\infty}^{(2)})$ (but not necessarily
$\bm{x}_{\infty}^{\left(  1\right)  }=\bm{x}_{\infty}^{(2)}$).

Now, we can proceed the process sufficiently often as%
\[
\frac{\nabla^{T}\mu_{j}(\bm{x}_{1})\Delta\bm{x}_{1}}{\mu_{j}(\bar{\bm{x}}%
)}\leq\frac{\nabla^{T}\mu_{j}(\bm{x}_{1})\Delta\bm{x}_{1}}{\mu_{j}%
(\bm{x}_{1})\left(  1+\frac{\nabla^{T}\mu_{j}(\bm{x}_{2})\Delta\bm{x}_{2}}%
{\mu_{j}(\bm{x}_{1})}\right)  }\leq\;...
\]
such that in the final step%
\begin{align*}
\frac{\nabla^{T}\mu_{j}(\bm{x}_{n+1})\Delta\bm{x}_{n+1}}{\mu_{j}(\bm{x}_{n})}
  =\frac{(\nabla^{T}\mu_{j}(\bm{x}_{\infty}^{\left(  k\right)  }%
)+\epsilon_{n}^{k})\Delta\bm{x}_{n+1}}{\mu_{j}(\bm{x}_{\infty}^{\left(
l\right)  })+\epsilon_{n}^{l}}
  =\frac{(\nabla^{T}\mu_{j}(\bm{x}_{\infty}^{\left(  k\right)  }%
)+\epsilon_{n}^{k})\Delta\bm{x}_{n+1}}{\mu_{j}(\bm{x}_{\infty}^{\left(
k\right)  })+\epsilon_{n}^{l}}
  \leq\epsilon,\;k,l=1,2,
\end{align*}
by condition (\ref{eq:prfCond1}). Hence, we have
\[
\frac{\sum_{j\in\mathcal{M}}\mu_{j}(\bar{\bm{x}})\left(  1+\frac{\nabla^{T}%
\mu_{j}(\tilde{\bm{x}})\Delta\bar{\bm{x}}}{\mu_{j}(\bar{\bm{x}})}\right)
}{\sum_{j\in\mathcal{M}}\mu_{j}(\bar{\bm{x}})\left(  1+\frac{\nabla^{T}\mu
_{j}(\tilde{\underline{\bm{x}}})\Delta\underline{\bm{x}}}{\mu_{j}(\bar
{\bm{x}})}\right)  }=\frac{\left(  1+\epsilon_{n}^{\prime}\right)  }{\left(
1+\epsilon_{n}^{\prime\prime}\right)  }=1+\epsilon_{n}^{\prime\prime\prime},
\]
$\epsilon_{n}^{\prime},\epsilon_{n}^{\prime\prime}\text{ zero sequences}$, and
further
$\left\vert \mu_{i}(\bm{x}+\Delta\bm{x})-\mu_{i}(\bm{x})(1+\epsilon
_{n}^{\prime\prime\prime})\right\vert \leq\left\vert \mu_{i}(\bm{x}+\Delta
\bm{x})-\mu_{i}(\bm{x})\right\vert +\mu_{i}(\bm{x})\epsilon_{n}^{\prime
\prime\prime}\leq\epsilon\sum_{j\in\mathcal{M}}\mu_{j}(\tilde{\bm{x}}%
)(1+\epsilon_{n}^{\prime})$,
which is equivalent to:%
\[
\left\vert \mu_{i}(\bm{x}+\Delta\bm{x})-\mu_{i}(\bm{x})\right\vert
\leq\epsilon\sum_{j\in\mathcal{M}}\mu_{j}(\bar{\bm{x}})(1+\epsilon_{n}%
^{\prime})-\epsilon_{n}^{\prime\prime}\mu_{i}(\bm{x}).
\]
Since $\bar{\bm{x}}$ is arbitrary and can be suitably choose, condition
(\ref{eq:prfCond1}) with some $\Vert\bm{x}\Vert>C_{6}\left(  \epsilon
^{\prime\prime\prime\prime}\right)  $ is sufficient for the latter to hold.

\section{Proof of Theorem \ref{thm:corollary2}} \label{sec:Appx2}
We can write
\[
\mu_{i}(\bm{x})=\frac{\partial h}{\partial x_{i}}(\bm{x})=l(x_{i}%
)\frac{\partial h_{0}}{\partial\tilde{x_{i}}}(\tilde{\bm{x}})
\]
where we defined $l:=\frac{\partial\tilde{x_{i}}}{\partial x_{i}}$. Note, that
here $\tilde{x_{i}}$ only depends on $x_{i}$. The gradient of the weight
$\mu_{i}(\bm{x})$ is given by:%
\[
\frac{\partial\mu_{i}}{\partial x_{j}}(\bm{x})=%
\begin{cases}
\frac{\partial l}{\partial x_{i}}(x_{i})\frac{\partial h_{0}}{\partial
\tilde{x_{i}}}(\tilde{\bm{x}})+l(x_{i})\frac{\partial}{\partial x_{i}}%
\frac{\partial h_{0}}{\partial{\tilde{x_{i}}}}(\tilde{\bm{x}}) & i=j\\
\frac{\partial}{\partial x_{j}}\frac{\partial h_{0}}{\partial\tilde{x}_{i}%
}(\tilde{\bm{x}})\cdot l(x_{i}) & i\neq j
\end{cases}
\]
Define $\bm{x}^{\Delta}:=\bm{x}+\Delta\bm{x}$ and $\tilde{\bm{x}}^{\Delta
}:=\tilde{\bm{x}}(\bm{x}^{\Delta})$. From the proof of Theorem
\ref{thm:corollary1} it is clear that we only have to show that
\[
\frac{\left\vert \nabla^{T}\mu_{i}(\bm{x})\Delta\bm{x}\right\vert }{\left\Vert
\boldsymbol{\mu}(\bm{x}^{\Delta})\right\Vert }\leq\epsilon,
\]
for some $\epsilon>0$ arbitrarily small. This can be rewritten as:%
\begin{gather*}
\frac{\frac{\partial l}{\partial x_{i}}(x_{i})\frac{\partial h_{0}}%
{\partial\tilde{x_{i}}}(\tilde{\bm{x}})\Delta x_{i}+l(x_{i})\frac{\partial
}{\partial x_{i}}\frac{\partial h_{0}}{\partial{\tilde{x_{i}}}}(\tilde
{\bm{x}})\Delta x_{i}}{\sum_{j\in\mathcal{M}}l(x_{j}^{\Delta})\frac{\partial
h_{0}}{\partial\tilde{x_{j}}}(\tilde{\bm{x}}^{\Delta})}
+\frac{l(x_{i})\sum_{j\in\mathcal{M},j\neq i}\frac{\partial}{\partial x_{j}%
}\frac{\partial h_{0}}{\partial\tilde{x}_{i}}(\tilde{\bm{x}})\Delta x_{j}%
}{\sum_{j\in\mathcal{M}}l(x_{j}^{\Delta})\frac{\partial h_{0}}{\partial
\tilde{x_{j}}}(\tilde{\bm{x}}^{\Delta})}\leq\epsilon
\end{gather*}
Since $\frac{\partial h_{0}}{\partial\tilde{x_{i}}},l$ are Lipschitz, thus
$\frac{\partial}{\partial x_{j}}\frac{\partial h_{0}}{\partial\tilde{x_{i}}%
},\frac{\partial l}{\partial x_{i}}$ are uniformly bounded, and $l(x_{i}%
),\frac{\partial h_{0}}{\partial\tilde{x}_{i}}(\tilde{\bm{x}})\geq
l^{1+\epsilon}(x_{i})\rightarrow\infty$ when $x_{i}\rightarrow\infty$, the
effect of $\Delta\bm{x}$ vanishes in the denominator. The condition
$\frac{\partial h_{0}}{\partial\tilde{x}_{i}}(\tilde{\bm{x}})\geq
l^{1+\epsilon}(x_{i})$ is required since we have expressions of the form%
\[
\frac{l(x_{i})l(x_{j})}{l(x_{i})\frac{\partial h_{0}}{\partial\tilde{x_{i}}%
}(\tilde{\bm{x}})+l(x_{j})\frac{\partial h_{0}}{\partial\tilde{x_{j}}}%
(\tilde{\bm{x}})}%
\]
which then become arbitrarily small.

\bibliographystyle{IEEEtran}
\bibliography{IEEEabrv,TUB_Report}

\end{document}